\documentclass[11pt]{article}
\usepackage[utf8]{inputenc}
\usepackage[left=1in, right=1in, bottom=1.15in, top=1.1in]{geometry}

\usepackage{amsmath,amssymb}
\usepackage{mathtools}
\usepackage[backref, colorlinks,citecolor=blue,linkcolor=magenta,bookmarks=true]{hyperref}

\usepackage[vlined,ruled,linesnumbered]{algorithm2e} 
\SetKwRepeat{Do}{do}{while}

\usepackage{verbatim} %for comment
\usepackage{multirow}
\usepackage{amsthm} %for *proof*
\newtheorem{theorem}{Theorem}[section]
\newtheorem{corollary}[theorem]{Corollary}

\newtheorem{proposition}[theorem]{Proposition}

\newtheorem{claim}[theorem]{Claim}
\usepackage{threeparttable}
\usepackage{graphicx}
\usepackage{subfigure}
\usepackage{tikz}
\usetikzlibrary{automata, positioning, arrows}
\tikzset{
box/.style ={
rectangle, % the shape 
rounded corners =1pt, % round corners
minimum width =50pt, %minimum width
minimum height =20pt, %minimum height
inner sep=5pt, %distance between the frame and the words
draw=black, % frame color}
align = center
}}

\title{Insightful Mining Equilibria}

\author{
Mengqian Zhang \\\textit{Shanghai Jiao Tong University} \\ \texttt{mengqian@sjtu.edu.cn}
\and Yuhao Li \\\textit{Columbia University}\\ \texttt{yuhaoli@cs.columbia.edu}
\vspace{0.15cm}
\and Jichen Li \\\textit{Peking University}\\ \texttt{limo923@pku.edu.cn}
\and Chaozhe Kong \\\textit{Peking University}\\ \texttt{kcz@pku.edu.cn}
\and Xiaotie Deng \\\textit{Peking University}\\ \texttt{xiaotie@pku.edu.cn}
\vspace{0.3cm}
}

\date{}

\begin{document}

\maketitle
\begin{abstract}
The selfish mining attack, arguably the most famous game-theoretic attack in blockchain, indicates that the Bitcoin protocol is not incentive-compatible. Most subsequent works mainly focus on strengthening the selfish mining strategy, thus enabling a single strategic agent more likely to deviate.
In sharp contrast, little attention has been paid to the resistant behavior against the selfish mining attack, let alone further equilibrium analysis for miners and mining pools in the blockchain as a multi-agent system.

In this paper, first, we propose a strategy called insightful mining to counteract selfish mining. By infiltrating an undercover miner into the selfish pool, the insightful pool could acquire the number of its hidden blocks. We prove that, with this extra insight, the utility of the insightful pool could be strictly greater than the selfish pool’s when they have the same mining power. 
Then we investigate the mining game where all pools can either choose to be honest or take the insightful mining strategy. We characterize the Nash equilibrium of this mining game, and derive three corollaries: (a) each mining game has a \textit{pure} Nash equilibrium;  (b) honest mining is a Nash equilibrium if the largest mining pool has a fraction of mining power no more than 1/3; (c) there are at most two insightful pools under equilibrium no matter how the mining power is distributed.

\end{abstract}

\thispagestyle{empty}
\newpage 
\setcounter{page}{1}

\section{Introduction}\label{sec:intro}
% What a mining pool can do if it is insightful?
% Blockchain, Bitcoin
% SM is not incentive-compatible
% Even if the honest miners have the same mining power, the revenue is quite smaller. with a figure in intro: when alpha(SM)=beta(honestly), the relative revenue of two players
% We propose a strategy, insightful, only more revenue.
% Our analysis is based on sophisticated Markov Reward Process 

% ending a new point: SM multi-agent 的生态下会演化成什么样

% We shed a new light that
% the first step 互相探测没有意义，重新回到honest mining

% 落脚点：抛出新的问题，counter measure
% bargaining set

% PoW重要性

% 一句话和PoS的关系？
% Mining的影响

% 被攻击了（Selfish）

% Insightful

% Our contributions（第一页）

% Simulation，方法？强调完整

% Related Work（可以包括反攻击）

% 用insight 把整篇paper穿起来

% 如果从selfish mining开始，要和intro consistent

Bitcoin~\cite{nakamoto2008bitcoin}, as the pioneering blockchain ecosystem, proposes an electronic payment system without any trusted party. %\textit{Transactions} in the system are recorded in a distributed \textit{ledger}, which is in the form of chain of \textit{blocks}. The blocks are generated by nodes (or say \textit{miners}). 
%To ensure security, Bitcoin 
It creatively uses \textit{Proof-of-Work} (PoW) to incentivise all \textit{miners} to solve a cryptopuzzle (also known as \textit{mining}). The winner will gain the record-keeping rights to generate a \textit{block} and be awarded the newly minted tokens. As more and more computational power is invested into mining, it may take a sole miner months or even years to find a block~\cite{schrijvers2016incentive}. In order to reduce the uncertainty, a group of miners form a \textit{mining pool} to share their computational resources. Under the leadership of the pool manager, all miners in a pool solve the same puzzle in parallel and share the block rewards. In the Bitcoin system, so long as all participants behave honestly, one's expected revenue will be proportional to its hashing power.

However, in practice, miners are rational and may act strategically. Thus, game theory naturally stands out as a tool to analyze the robustness of the Bitcoin protocol.
The conventional wisdom would expect a proof of the incentive compatibility of the Bitcoin protocol and subsequently the strategyproofness against manipulative miners.
%from the perspective of game theory.
% Yuhao!!!
% as a tool

% A major line is selfish mining.

Such a hope was broken by the seminal work~\cite{eyal2014majority}, which proposed the selfish mining strategy, arguably the most well-known game-theoretic attack in blockchain. It indicates that the Bitcoin mining protocol
is not incentive-compatible. 
% \textcolor{red}{Regarding honest miners as the level-zero player, the selfish mining can be seen as a response by the level-one player.} It works because of the different levels of rationality, which gives the selfish pool some competitive advantages. Specifically, the selfish pool is suggested to keep its blocks secret and reveal them strategically when its mining power is above a certain threshold size. 
The key idea behind the attack is to induce honest miners to waste mining power of honest miners.
% on the branch destined to fail. 
As a result, the selfish pool could obtain more revenue than its fair share. 

Pushing this approach to the extreme,
Sapirshtein et al.~\cite{sapirshtein2016optimal} expanded the action space of selfish mining, modeled it as a Markov Decision Process (MDP), and pioneered a novel technique to resolve the non-linear objective function of the MDP to get a more powerful selfish mining strategy, for a revenue arbitrarily close to the optimum.
A series of works have since been initiated to study the mining strategies of a rational pool under the same  assumption that other pools behave honestly~\cite{nayak2016stubborn,feng2019selfish,grunspan2020selfish,ritz2018impact,li2021new,negy2020selfish}. 
% Another milestone 

In sharp contrast, little attention has been paid to the incentive of the victims, which plays an important role in studying the strategic interactions among participants and understanding the stable state of blockchain systems. In this paper, we propose and study the following vital questions.

%\vspace{0.8em}{\noindent

\begin{center}
\textit{%What 
Can a pool strategically defend %facing 
against the selfish mining attack?\\ And what equilibrium will the ecosystem of different types of agents reach eventually?}
\end{center}

%how will the ecosystem \textcolor{red}{evolve} subsequently?}}

%what equilibrium of the ecosystem will reach under the strategic interaction?
%In sharp contrast, there has been much less attention on the vital problem of \textit{what a pool should do when it knows a pool is selfish mining}, in other words, the incentives of a strategic pool when there is one selfish mining pool.

\subsection{Our Contributions}
In this work, we propose a  strategy called \textit{insightful mining} (Fig. \ref{fig:algorithm}) and call the pool adopting the insightful mining strategy as the \textit{insightful pool}. Once detecting a selfish pool, an insightful pool can infiltrate an undercover miner into it to monitor the number of hidden blocks.\footnote{We discuss this action in more detail in Section \ref{sec:strategy}.} With this key information, the insightful pool keeps a clear view of the current mining competition, \textit{i.e.}, the length of the public branch and the selfish branch of the blockchain. Then it utilizes this information to respond strategically. At a high level view, when observing that the selfish pool is taking the lead, the insightful pool would behave honestly to end 
its leading advantage as quickly as possible. On the other hand, 
when the insightful pool is taking the lead, it will execute a strategy similar to selfish mining, regarding the selfish pool and the honest pool as ``others''.

We model the intersections among the honest pool, the selfish pool, and the insightful pool as a two-dimensional Markov reward process with an infinite number of states (Fig.~\ref{fig:transitions} and Table~\ref{tab:transition}). We prove that, when there is a selfish pool and an insightful pool with the same mining power, the insightful pool would get a strictly greater expected revenue than the selfish pool (Theorem~\ref{thm:domination}). This demonstrates that the extra insight significantly reverses the selfish pool's advantage.
% to benefit the countermeasure against the selfish miners.
%Compared to the case where there is only one selfish pool and the others are honest pools, this result shows the extra insight could reverse the selfish pool's advantage over the honest pools significantly to benefit the countermeasure against the selfish miners.

% insight 过渡?

% To more deeply understand the selfish mining strategy with multiple rational pools, equilibrium analysis of such a multi-agent system is critical. However, the complications involved in the scenarios with multiple selfish pools made progress in this direction slow, since the miners may keep an arbitrarily long private chain.
%加引用or加解释、换地方？
% It is well-known that the selfish mining rules out following the Bitcoin protocol truthfully as a Nash equilibrium. We further observe that the scenario with exactly one selfish pool is also not a Nash equilibrium, as the honest pools may have incentives to adopt other strategies (\textit{e.g.}, our insightful mining). Nevertheless, it is no exaggeration to say that characterizing Nash equilibria is still one of the most important questions in the research line of selfish mining.

% Insightful as a single strategy
% To deeply understand the insightful mining strategy,
Then we investigate the scene where all $n$ mining pools are strategic.
Besides counteracting the selfish mining attack, insightful mining can be adopted directly as a mining strategy. Specifically, insightful mining resembles selfish mining if there is no pool mining selfishly. We study the mining game where each pool plants spies into all other pools, and chooses either to follow the Bitcoin protocol or to take the insightful mining strategy. Such a mining game can be formulated as an $n$-player normal-form game. Note that although there are $2^n$ pure strategy profiles, the payoff function of each player is explicitly represented (Proposition \ref{proposition: ER}). Our main result here is a characterization theorem of the Nash equilibrium in mining games (Theorem \ref{thm:equilibria}). 
Concretely, Theorem \ref{thm:equilibria} derives three corollaries: (a) each mining game has a \textit{pure} Nash equilibrium;  (b) honest mining is a Nash equilibrium if the largest mining pool has a fraction of total hashing power no more than 1/3; (c) there are at most two insightful pools under equilibrium no matter how the mining power is distributed.

Beyond our theoretical results, we also conduct a series of simulations to understand insightful mining. First, we explore the performance of the insightful mining strategy when the insightful pool and the selfish pool have different mining power (Section \ref{sec: Insightful Mining vs. Selfish Mining}). Simulation results provide compelling evidence that the insightful pool could gain more revenue even if it holds less mining power than the selfish pool. Second, we 
learn an optimal insightful mining strategy via Markov Decision Process (Section \ref{sec: Optimal Insightful Mining}). It tells us that once the insightful branch is overridden during the mining competition, sticking to it for a while is a better choice than giving up immediately. Finally, we focus on the mining game and visualize the participants' equilibrium strategy, as well as the corresponding utility  (Section \ref{sec: Equilibria Visualization}).

In the end, we discuss the role of the undercover miner in the context of selfish mining and blockchain, which sheds new light upon the future research directions (Section \ref{sec:discussion}).
%In terms of discussion and future work (\Cref{sec:discussion}), beyond our insightful mining strategy, we open up some future directions based on the undercover miner in a more broad blockchain context.

We summarize our main contributions in the following Fig. \ref{tab:Summary of our contributions}.

\begin{figure}[!h]
    \centering
    \includegraphics[width=0.9\textwidth]{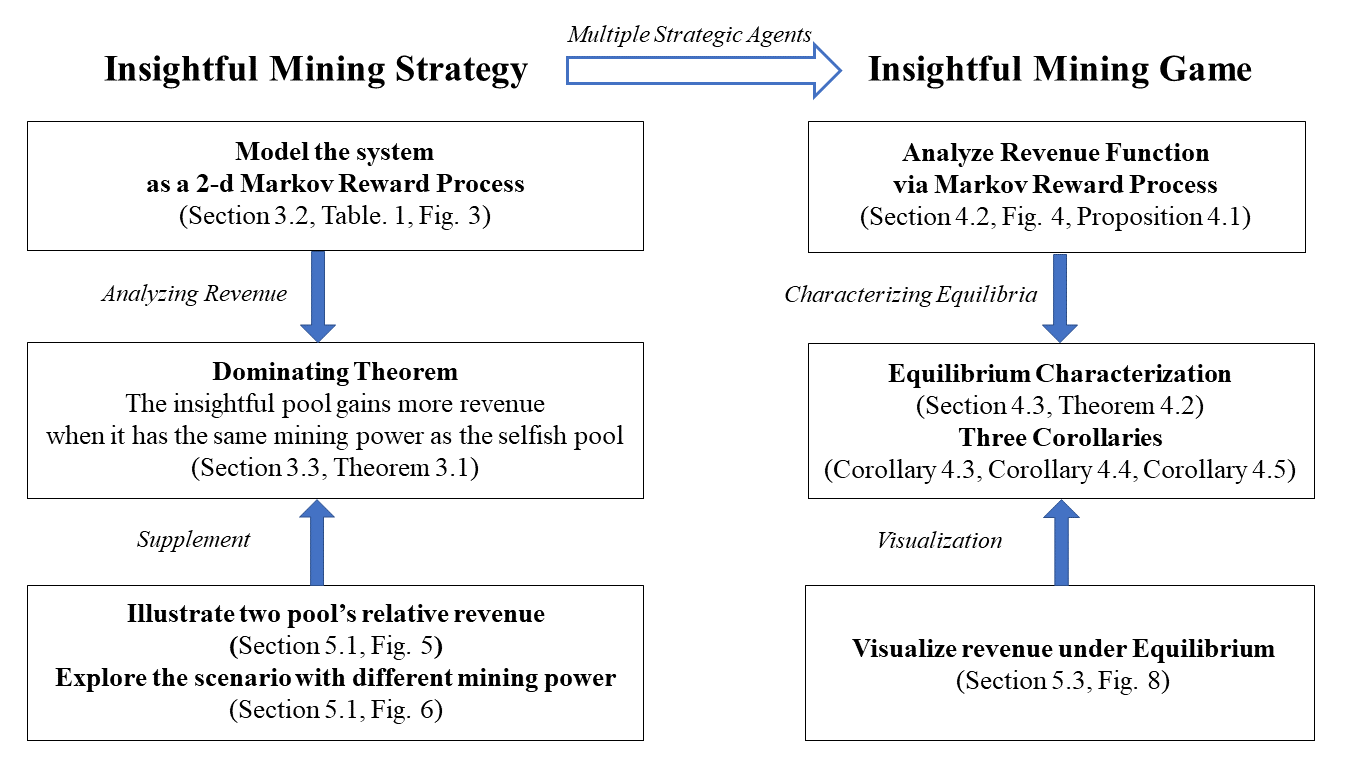}
    \caption{Summary of our contributions.}
    \label{tab:Summary of our contributions}
\end{figure}

\subsection{Related Work}

The classic selfish mining attack was first proposed and mathematically modeled as a Markov reward process in the seminal paper~\cite{eyal2014majority}. Observing that the classic selfish mining strategy could be suboptimal for a large parameter space, several works~\cite{sapirshtein2016optimal,nayak2016stubborn} further generalized the system as a Markov Decision Process (MDP) to find the optimal selfish mining strategy. Aiming to solve the average-MDP with a non-linear objective function, \cite{sapirshtein2016optimal} proposed a binary search procedure by converting the problem into a series of standard MDPs, which is also adopted for learning our \textit{optimal} insightful mining strategy in Section \ref{sec: Optimal Insightful Mining}. A recent work~\cite{zur2020efficient} developed a more efficient method called Probabilistic Termination Optimization, converting the average-MDP into only one standard MDP. %Like all these previous works, in this paper, we also assume that the time to broadcast a block is negligible, and the transaction fee is negligible. In other words, miners' revenue mainly comes from block rewards.

It was more challenging to study other agents' incentives against one selfish miner, due to the tremendous state spaces and complicated Markov reward processes. The work of~\cite{liu2018strategy} presented some simulation results on systems involving multiple selfish miners~\cite{eyal2014majority} or involving multiple stubborn miners~\cite{nayak2016stubborn}. On the learning side, a recent work \cite{hou2019squirrl} proposed a novel framework called SquirRL, which is based on deep reinforcement learning (deep-RL) techniques. Their experiments suggest that when facing selfish mining, adopting selfish mining might not be the optimal choice. We \textit{prove} such a result by providing the insightful mining strategy and the dominating theorem (Theorem~\ref{thm:domination}). The most strengthen of SquirRL is a more general strategy space generated by deep-RL. We highlight that it cannot cover our insightful mining strategy since our most strengthen comes from our undercover miner's insights (information), which was not discussed in the wide selfish mining context.

To our best knowledge, the most related work that theoretically studied the equilibria with multiple selfish mining pools is~\cite{marmolejo2019competing}. Due to the analytical challenges of infinite states in the classic selfish mining strategy, they proposed a simplified version called semi-selfish mining, where the strategic mining pool will only keep a private chain of length at most two. Such a restriction makes the Markov reward process have finite states (as long as there are finite number of semi-selfish miners) and subsequently simplifies the equilibrium analysis. However, our insightful mining strategy works against the classic selfish mining strategy, and itself may also keep an arbitrary long private chain. While this leads to a 2-dimensional Markov reward process with an infinite number of states, the techniques in mathematical analysis are sufficient for us to prove the desired dominating theorem (Theorem~\ref{thm:domination}) and equilibrium characterization (Theorem~\ref{thm:equilibria}).
% and makes it harder to solve the exact formula of relative revenues, we do a more careful  over the comparison between selfish mining and insightful mining, and the equilibrium analysis.

% Some research like~\cite{zhang2017publish,sompolinsky2015secure} proposed some method to defence selfish attack, but they all change the basic mining protocol. Research work~\cite{feng2019selfish} also explains the universality of selfish attack in the blockchain system.

% \subsection{Organization}
% We organize the rest of the paper as follows. In~\Cref{sec:background} we introduce the background of blockchain and the Proof-of-Work framework, and explain the detailed reason of the ability to plant an undercover miner. In~\Cref{sec:insightful_mining}, we formally introduce the insightful mining strategy, model it as a 2-dimensional Markov reward process, and prove the dominating theorem. In~\Cref{sec: Mining Game and Equilibria}, we formulate the mining game and study its equilibria. In~\Cref{sec:simulation}, we further conduct more simulation results to show the optimal insightful mining and visualize pools' revenues. Finally, in~\Cref{sec:discussion}, beyond our insightful mining strategy, we open up some future directions based on the undercover miner in a more broad blockchain context.

\section{Preliminaries}\label{sec:background}

\subsection{Proof of Work}
In the context of blockchain, Proof of Work was first introduced in Bitcoin~\cite{nakamoto2008bitcoin}. As mentioned, the security of Bitcoin heavily relies on the Proof-of-Work scheme, which has also been widely adopted by other blockchain systems like Ethereum~\cite{buterin2014next}. The past decade has seen a great amount of research around PoW, with respect to its block rewards design~\cite{chen2019axiomatic}, strategic deviation~\cite{kiayias2016blockchain}, the difficulty adjustment algorithm~\cite{noda2020economic,goren2019mind}, the energy costs~\cite{fiat2019energy} and so on. 

Taking the Bitcoin as an example, PoW requires a miner to randomly engage in the hashing function calls to solve a cryptopuzzle. Typically, miners should search for a nonce value satisfying that
\begin{equation}
    H(previous \ hash; \ address; \ Merkle \ root; \ nonce) \leq D
    \label{eq:puzzle}
\end{equation}
where $H(\cdot)$ is a commonly known hash function (\textit{e.g.}, SHA-256 in Bitcoin); \textit{precious hash} is the hash value of the previous block; \textit{address} is the miner's address to receive potential rewards; \textit{Merkle root} is an integrated hash value of all transactions in the block; and $D$ is the target of the problem and reflects the difficulty of this puzzle.\footnote{For security, the difficult of puzzles will be adjusted automatically, to make sure the mean interval of
block generation is 10 minutes.} Started from the genesis block, all miners compete to find a feasible solution, thus generating a new block appended to the previous one. In return, they will be awarded the newly minted bitcoins for their efforts in maintaining the blockchain system. The standard Bitcoin protocol treats the longest chain as the main chain. Once encountering two blocks at the same block height, miners randomly choose one to follow according to the uniform tie-breaking rule. Thus, in order to be accepted by more miners, it is suggested to publish the newly generated block immediately. In this paper, the miners who sticks to the Bitcoin protocol are referred to be \textit{honest}.

\subsection{Mining Pool}
\label{sec: reason of spy}
With more and more hashing power invested into the Bitcoin mining, the chances of finding a block as a sole miner are quite slim. Nowadays, miners tend to participate in organizations called mining pool.

In general, a mining pool is composed of a pool manager and several peer miners. All participants shall cooperate to solve the same puzzle. Specifically, each miner will receive a task like (\ref{eq:puzzle}) above from the pool manager, as well as a work unit containing a particular range of nonce. Instead of trying all possible nonce values, the miner only needs to search the answer from the received work unit. In this way, all miners in the pool work in parallel. Once any miner find a valid solution, this pool succeeds in this mining competition. Then a new task will be organised and further released to all miners in the pool. Also, participants will share the mining rewards according to the reward allocation protocol like Pay Per Share (PPS), proportional (PROP), Pay Per Last N Shares (PPLNS)~\cite{zolotavkin2017incentive} and so on. In expectation, the miners' rewards is proportional to their hashing power. As a result, miners who join the mining pool can significantly reduce the variance of mining rewards. Currently, most of the blocks in Bitcoin are generated by mining pools such as AntPool~\cite{Antpool}, Poolin~\cite{Poolin}, F2Pool~\cite{gencer2018decentralization}.

\subsection{Selfish Mining}
It has long been believed that Bitcoin protocol is incentive-compatible. However, Eyal and Sirer~\cite{eyal2014majority} indicates that it is not the case. It describes a well-known attack called selfish mining. A pool could receive higher rewards than its fair share via the selfish mining strategy. This attack exploits ingeniously the conflict-resolution rule of the Bitcoin protocol, in which when encountering a fork, only one chain of blocks will be considered valid. With the selfish mining strategy, the attacker deliberately creates a fork and forces honest miners to waste efforts on a stale branch. Specifically, the selfish pool strategically keep its newly found block in secret, rather than publishing it immediately. Afterwards, it continues to mine on the head of this private branch. When the honest miners generate a new block, the selfish pool will correspondingly publish one private block at the same height and thus create a fork. Once the selfish pool's leads reduces to two, an honest block will prompt the selfish pool to reveal all its private branch. As a well-known conclusion, assuming that the honest miners apply the uniform tie-breaking rule, if the fraction of the selfish pool's mining power is greater than $25\%$, it will always get more benefit than behaving honestly.

\section{Insightful Mining Strategy}\label{sec:insightful_mining}
\subsection{Model and Strategy}\label{sec:strategy}

This paper considers a system of $n$ miners. Each miner $i$ has $m_i$ fraction of total hashing power, such that $\sum_{i=1}^n m_i = 1$. Let $\mathcal{H}$, $\mathcal{S}$, $\mathcal{I}$ denote the set of honest miners, selfish miners, and insightful miners respectively. As the honest miners strictly follow the Bitcoin protocol and do not hide any block information from each other, they are regarded as a whole, which is referred to as the \textit{honest pool} in the paper. Similarly, all selfish miners who adopt the selfish mining strategy combine together to behave as a single agent, which is called the \textit{selfish pool}. The remaining miners form the \textit{insightful pool} and adopt the insightful strategy stated later. Let $\alpha$ and $\beta$ denote the fraction of mining power controlled by the selfish pool and the insightful pool respectively. We have $\alpha=\sum_{i \in \mathcal{S}}m_i$ and $\beta=\sum_{i \in \mathcal{I}}m_i$. Then the total power of the honest pool can be represented as $1-\alpha-\beta$. Following the previous work~\cite{eyal2014majority,sapirshtein2016optimal}, in this paper, we also assume that the time to broadcast a block is negligible, and the transaction fee is negligible. In other words, the pools' revenue mainly comes from block rewards. In addition, the block generation is treated as a randomized model, where a new block is generated in each time slot.

Now we describe the insightful mining strategy. Before getting into the details, we state that the insightful pool could learn how many blocks the selfish pool has been hiding by doing the following. The manager of the insightful pool shall pretend to join the selfish pool as a spy. As a pool member, it will receive a mining task from the manager of the selfish pool. The hash value of the previous block can be parsed from the task. Normally, this hash value corresponds to the last block of the main chain. Once the selfish pool mines a block\footnote{A member of the selfish pool finds an acceptable nonce to the cryptopuzzle and submits it to the manager.}, its manager will keep the block privately and publish a new task based on it. From the perspective of the spy, there is no newly published block in the system, but the selfish manager releases a new task which is based on an unknown block. Then it is reasonable to believe that the selfish manager is hiding blocks. Furthermore, the number of hidden blocks is exactly the number of recently received tasks with unmatched \textit{previous hash}. 

By working as a spy\footnote{We assume that the mining power of this spy is negligible, as well as its revenue from the selfish pool.}, the insightful pool has a clear understanding of the system's situation, \textit{i.e.}, the mining progress of each player. Although all pools are mining after the main chain, the three players may hold different sub-chain (also referred to as \textit{branch}) during the mining competition. Let $l_h$, $l_s$, $l_i$ denote the length of honest branch, selfish branch, and insightful branch respectively. In the process of mining, the honest pool only knows the public information $l_h$. The selfish pool is aware of both $l_h$ and $l_s$, while the insightful pool can observe all three lengths. Then the three types of players compete to generate blocks based on their own information. Their competition works in rounds. Each round begins with a global consensus on the current longest chain. When the selfish pool and insightful pool reveal all their private blocks, or they have no hidden blocks while the honest pool finds a block (see \textit{Case 1} below), the round ends, leading to a new global consensus. For the first block in a round, there are three possible cases.

% \begin{figure}
%     \centering
%     \includegraphics[width=0.96\textwidth]{Figure/case2.png}
%     \caption{An example of \textit{Case 2} in Section~\ref{sec:strategy}. After observing that the selfish pool keeps its block secretly, the insightful pool behaves honestly and mines on the opposite of the selfish branch, especially when the honest branch has the same length (see (c) above).}
%     \label{fig:case2}
% \end{figure}

\textit{Case 1: the honest pool generates the first block.} 
With probability $1-\alpha-\beta$, the honest pool mines a block and broadcasts it immediately. In this case, the insightful pool accepts this newly generated block and mines after it. According to the selfish mining strategy, the selfish pool will do the same. Consequently, all players reach consensus in this case and compete for the next block.

\textit{Case 2: the selfish pool generates the first block.}
With probability $\alpha$, the selfish pool mines a block. Based on the selfish mining strategy, the selfish pool will keep it private, aiming to further extend its lead. After observing this situation through the spy in the selfish pool, the insightful pool behaves honestly until the selfish pool reveals all its hidden blocks. Recall that when facing two branches of the same height, the honest pool uniformly chooses one of the them. The insightful pool, however, will deterministically mine on the opposite of the selfish branch. The key insight behind this strategy is to prompt the selfish pool to reveal all its hidden blocks and end its leading advantage as quickly as possible.

\textit{Case 3: the insightful pool generates the first block.}
With probability $\beta$, the insightful pool mines a block. It hides this block and takes the following actions which are similar to selfish mining. The insightful pool keeps a watchful eye on how many blocks the selfish pool and the honest pool have mined respectively. In the following competition, when its lead is larger than one (\textit{i.e.}, $l_i-\max\{l_h,l_s\}>1$), the insightful pool always hides all its mined blocks. Otherwise, it reveals the private branch all at once. Here, the way of releasing blocks is different from selfish mining which reveals blocks one by one in response to honest behavior. Both methods bring the same revenue. Releasing all blocks at the last moment, however, can make this strategic behavior avoid being detected.

Fig.~\ref{fig:algorithm} visualizes the flow chart of the insightful mining strategy.
%We have also shown the insightful mining strategy in Algorithm~\ref{alg:strategy}.  

\begin{figure}
    \centering
    \includegraphics[width=0.95\textwidth]{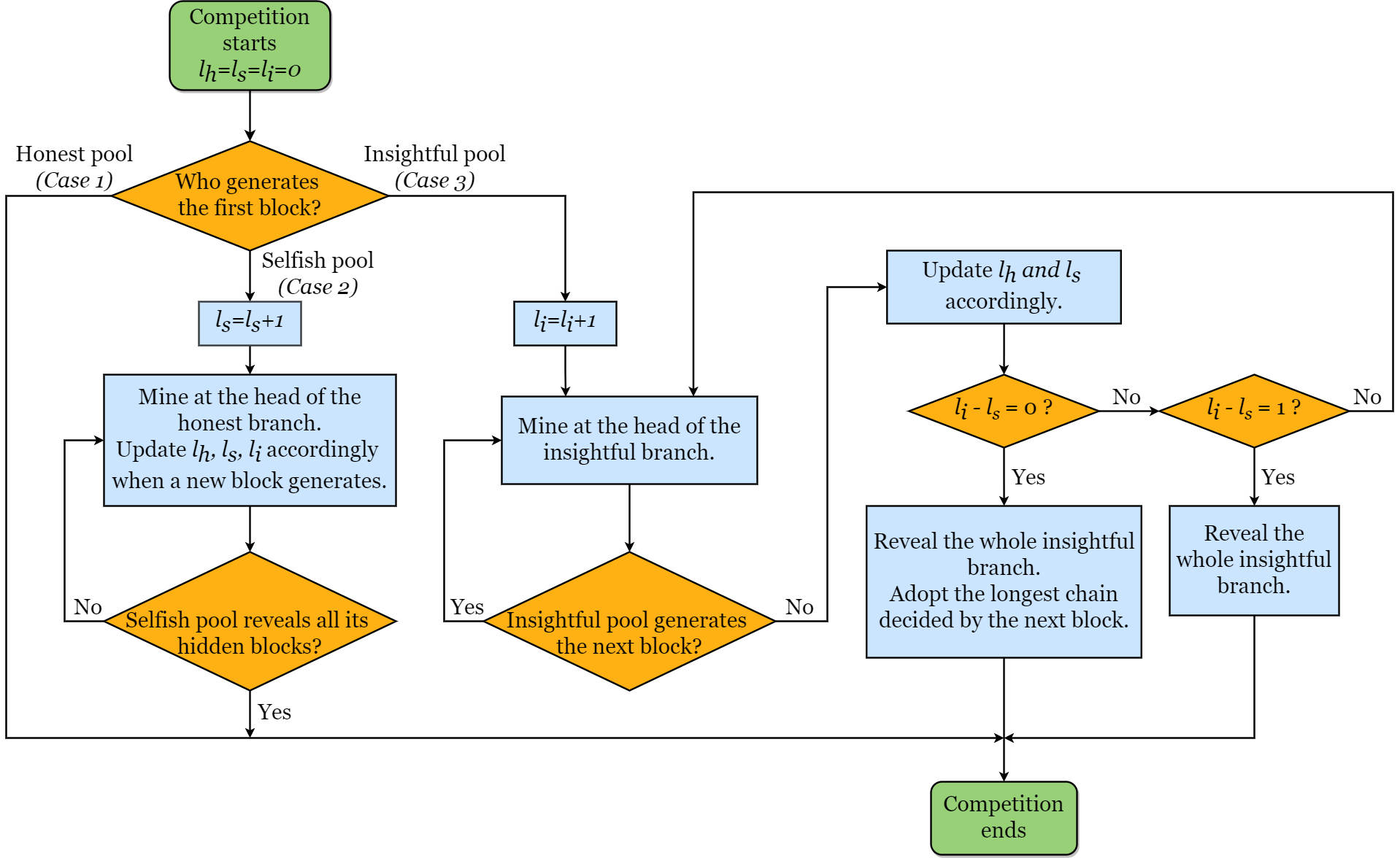}
    \caption{Flow chart of the insightful mining. $l_h$, $l_s$ and $l_i$ are the length of the honest branch, selfish branch, and insightful branch respectively.}
    \label{fig:algorithm}
\end{figure}

\subsection{Markov Reward Process}\label{sec: Markov Reward Process}
To analyse the relative revenue of different players under this strategy, we use a two-dimensional state $s = (x,y)$ to reflect the system status and further model the mining events as a Markov Reward Process. 
%Recall that $l_h$ is the number of blocks on the public chain which is maintained by honest miners. Moreover, $l_s$ and $l_i$ are the length of two branches hidden by the selfish pool and insightful pool. 
The state $x$ denotes the selfish pool's lead over the honest pool, \textit{i.e.}, the number of blocks that the selfish pool has not revealed. Similarly, $y$ is the insightful pool's lead over the selfish pool. Thus, we have $x, y \in \mathbb{N}\cup\{0'\}$ ($0'$ will be explained soon). Here, zero means the selfish pool (corresponding to $x$) or the insightful pool (corresponding to $y$) has no hidden blocks. Specifically, it contains two different states which we use $0$ and $0'$ to distinguish. Take $x$ as an example. The state $x=0$ indicates that the honest pool and the selfish pool are in agreement about a public chain. In other words, their branches are exactly the same. The state $x=0'$ means that the selfish pool and others (the honest pool or the insightful pool) hold a separate branch of the same length, and the selfish pool has revealed all blocks on its branch. In the state of $0'$, the next block will break the tie and decides the longest chain. For $y$, the meanings of state $0$ and $0'$ are similar to the above, with the insightful pool and others (the selfish pool and the honest pool) as two players.

\begin{table}[!t]
\caption{The state transitions and corresponding revenues.}\label{tab:transition}
\resizebox{0.95\textwidth}{!}{ 
\begin{tabular}{c|cclll}
\hline\hline
No. &  State $s$ & State $\tilde{s}$ & $Pr[s,\tilde{s}]$ & $r[s,\tilde{s}]$ & Conditions \\
\hline
1 & (0,0) & (0,0) & $1 - \alpha - \beta$ & $(1,0,0)$ & \\
2 & (0,0) & (1,0) & $\alpha$ & $(0,0,0)$    &   \\
3 & (1,0) & $(0',0)$ & $1 - \alpha - \beta$ & $(0,0,0)$ & \\
4 & $(0',0)$ & (0,0) & 1 & $(\frac{3-3\alpha-\beta}{2},\frac{1+3\alpha-\beta}{2},\beta)$ &    \\
5 & (1,0) & $(1,0')$ & $\beta$ & $(0,0,0)$ &   \\
6 & $(1,0')$ & (0,0) & 1 & $(1-\alpha-\beta,\frac{1+3\alpha-\beta}{2},\frac{1-\alpha+3\beta}{2})$ &  \\
7 & $(x,0)$ & $(x+1,0)$ & $\alpha$ & (0,0,0) & $\forall x \geq 1$    \\
8 & (2,0) & (0,0) & $1-\alpha$  &   (0,2,0) &   \\
9 & $(x,0)$   &   $(x-1,0)$  & $1-\alpha$  &    (0,1,0)  &   $\forall x \geq 3$  \\
10 & (0,0)  &   (0,1)   &   $\beta$ &   (0,0,0) &   \\
11  &   (0,1)   &   $(1,0')$    &    $\alpha$   &    (0,0,0)  &   \\
12  &   (0,1)   &   $(0,0')$    &   $1-\alpha-\beta$    &   (0,0,0) &   \\
13  &   $(0,0')$    &   (0,0)   &   1   &   $(\frac{3-2\alpha-3\beta}{2},\alpha,\frac{1+3\beta}{2})$ &   \\
14  &   (0,1)   &   (0,2)   &   $\beta$ &   (0,0,0) &   \\
15  &   (0,2)   &   (0,0)   &   $1-\beta$   &   (0,0,2)   &   \\
16  &   $(x,y)$ &   $(x,y+1)$   &   $\beta$ &   (0,0,0)   &   $\forall x\in \{0'\}\bigcup\mathbb{N}$, $y\geq 2$    \\
17  &   $(0,y)$ &   $(0,y-1)$   &   $1-\alpha-\beta$    &   (0,0,1) &   $\forall y \geq 3$   \\
18  &   $(x,y)$ &   $(x+1,y-1)$ &   $\alpha$    &    (0,0,1) &   $\forall x \geq 0$, $y \geq 3$  \\
19  &   $(1,y)$ &   $(0',y)$    &   $1-\alpha-\beta$    &   (0,0,0) &   $\forall y \geq 2$   \\
20  &   $(2,y)$ &   $(0,y)$    &   $1-\alpha-\beta$    &   (0,0,0) &   $\forall y \geq 2$   \\
21  &   $(x,y)$ &   $(x-1,y)$   &   $1-\alpha-\beta$    &   (0,0,0) &   $\forall 3 \geq 2$, $y \geq 2$   \\
22  &   $(x,2)$ &   (0,0) &   $\alpha$    &   (0,0,2)   &   $\forall x \geq 1$  \\
23  &   $(0',2)$    &   (0,0)   &   $1-\beta$   & (0,0,2) &   \\
24  &   $(0',y)$    &   $(0,y-1)$   &   $1-\beta$   & (0,0,1) & $\forall y \geq 3$  \\
\hline\hline
\end{tabular}
}
\end{table}

\begin{figure}[!ht]
\includegraphics[width=0.85\textwidth]{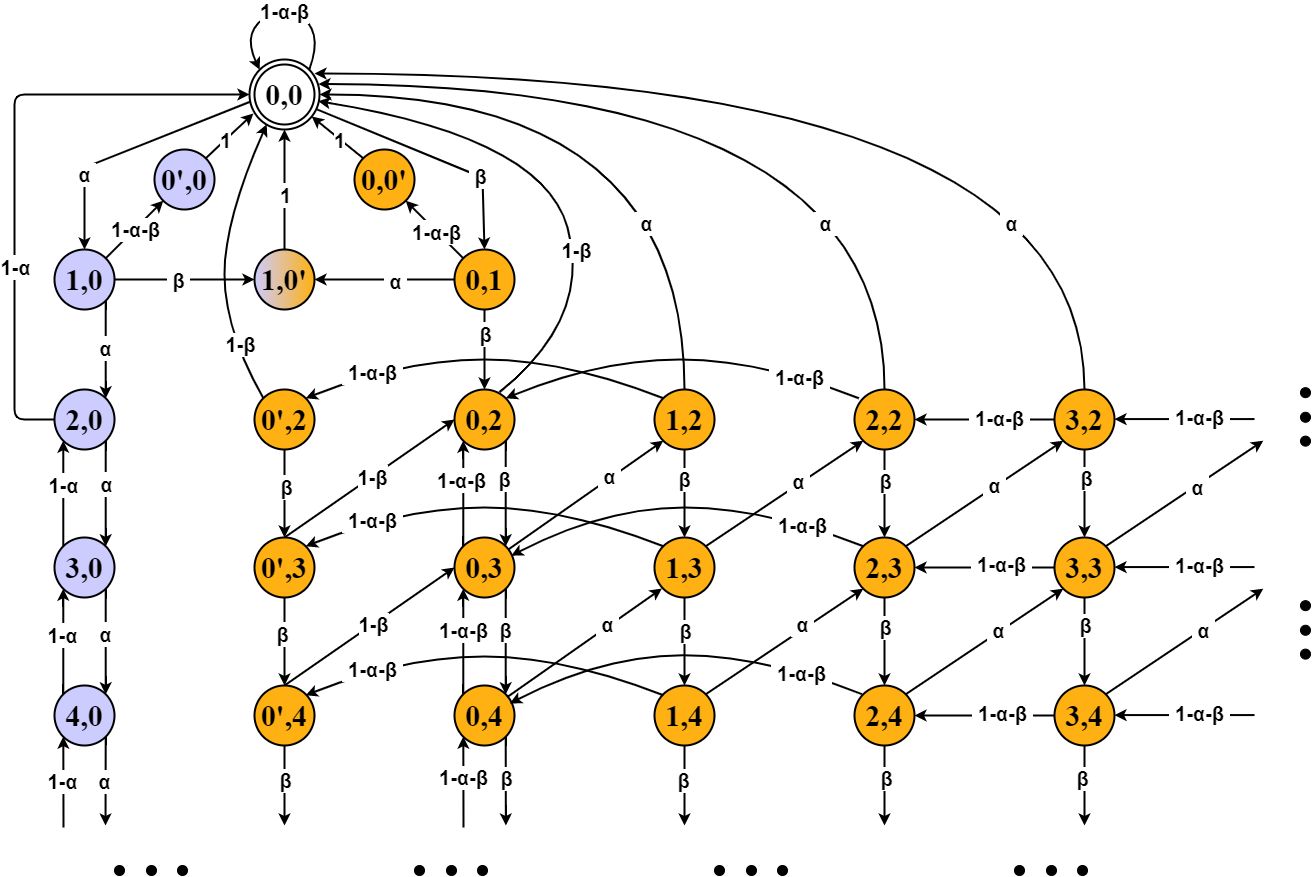}
\caption{The Markov Process of the system under the insightful mining strategy.} 
\label{fig:transitions}
\end{figure}

Let $Pr[s,\tilde{s}]$ denote the probability of changing from state $s$ to state $\tilde{s}$. The vector $r[s,\tilde{s}]$ represents the expected revenue obtained from this state transition. It contains three components, corresponding to the revenue of the honest pool, the selfish pool, and the insightful pool respectively. With the help of these notations, Table~\ref{tab:transition} lists the detailed state transitions and corresponding revenues in the system. Specifically, the item (1) formalizes Case 1 in Section~\ref{sec:strategy}. Items (2)-(9) correspond to Case 2, and Case 3 contains the items (10)-(24). Fig.~\ref{fig:transitions} also illustrates the overall state transitions in a more intuitive way. The detailed analysis of each transition can be found in Appendix~\ref{sec:appendixA}.

Recall that a branch will win at the end of one round. It is easy to verify that in our design each block of this winning branch will be awarded to some player once and only once.

\newcommand\yuhao[1]{\textcolor{cyan}{[Yuhao: #1]}}

%\mathbb
\newcommand{\bbA}{\mathbb{A}}
\newcommand{\bbB}{\mathbb{B}}
\newcommand{\bbC}{\mathbb{C}}
\newcommand{\bbD}{\mathbb{D}}
\newcommand{\bbE}{\mathbb{E}}
\newcommand{\bbF}{\mathbb{F}}
\newcommand{\bbG}{\mathbb{G}}
\newcommand{\bbH}{\mathbb{H}}
\newcommand{\bbI}{\mathbb{I}}
\newcommand{\bbJ}{\mathbb{J}}
\newcommand{\bbK}{\mathbb{K}}
\newcommand{\bbL}{\mathbb{L}}
\newcommand{\bbM}{\mathbb{M}}
\newcommand{\bbN}{\mathbb{N}}
\newcommand{\bbO}{\mathbb{O}}
\newcommand{\bbP}{\mathbb{P}}
\newcommand{\bbQ}{\mathbb{Q}}
\newcommand{\bbR}{\mathbb{R}}
\newcommand{\bbS}{\mathbb{S}}
\newcommand{\bbT}{\mathbb{T}}
\newcommand{\bbU}{\mathbb{U}}
\newcommand{\bbV}{\mathbb{V}}
\newcommand{\bbW}{\mathbb{W}}
\newcommand{\bbX}{\mathbb{X}}
\newcommand{\bbY}{\mathbb{Y}}
\newcommand{\bbZ}{\mathbb{Z}}

\subsection{The Dominating Theorem}\label{sec:Dominating Theorem}

Let $M:=\{H,IM,SM\}$. The utility of each $i\in M$ is the relative revenue defined as follows:
$$RREV_{i}=\bbE\left[\liminf_{T\rightarrow \infty}\frac{\sum_{t=1}^T r_{i}[s_t,s_{t+1}]\Big|s_0=(0,0),s_{t+1}\sim Pr[s_{t},s_{t+1}]}{\sum_{t=1}^T\sum_{j\in M} r_{j}[s_t,s_{t+1}]\Big|s_0=(0,0),s_{t+1}\sim Pr[s_{t},s_{t+1}]}\right].$$
Note that the transition probability $Pr[s_t,s_{t+1}]$, $\forall i\in\{H,SM,IM\}$, $RREV_i$ and $r_i[s_t,s_{t+1}]$ should depend on the mining power $\alpha$ and $\beta$. Here we simplify the notation without ambiguity. We denote the Markov Reward Process of Fig.~\ref{fig:transitions} by $Markov(\alpha,\beta)$.

Like previous work~\cite{hou2019squirrl}, here we focus on the scenario where the selfish pool and the insightful pool have the same mining power. The following theorem asserts that, in this case, the expected revenue of the insightful pool is strictly greater than the the expected revenue of selfish pool. The scenario with different pool sizes (\textit{i.e.}, $\alpha \neq \beta$) will be explored in Section~\ref{sec: Insightful Mining vs. Selfish Mining}.

\begin{theorem}
\label{thm:domination}
Let $\alpha$ and $\beta$ be the fraction of mining power that selfish pool and the insightful pool controls respectively. When $0<\alpha=\beta<\frac{1}{2}$, $RREV_{SM}(\alpha,\beta)<RREV_{IM}(\alpha,\beta)$ holds.
\end{theorem}

\begin{proof}
First suppose that the stationary distribution of $Markov(\alpha,\beta)$ exists, denoted by $\pi$, where $\pi_{(i,j)}$ represents the stationary probability of state $(i,j)$. Let $\bbS$ denote the (countable) set of states in $Markov(\alpha,\beta)$.

Let $ER_{H}(\alpha,\beta),ER_{SM}(\alpha,\beta),ER_{IM}(\alpha,\beta)$ denote the expected reward of honest pool, selfish pool, and insightful pool respectively, which can be calculated as follows. 

For each $i\in\{H,SM,IM\}$, we have
$$ER_{i}(\alpha,\beta) = \sum_{s\in\bbS}\sum_{\tilde{s}\in\bbS} Pr[s,\tilde{s}]\cdot r_{i}[s,\tilde{s}]\cdot\pi_s,$$
where $r_{H}[s,\tilde{s}]$, $r_{SM}[s,\tilde{s}]$ and $r_{IM}[s,\tilde{s}]$ denote $r[s,\tilde{s}](0)$, $r[s,\tilde{s}](1)$, and $r[s,\tilde{s}](2)$ respectively.

Then we have $$RREV_{SM}(\alpha,\beta)=\frac{ER_{SM}(\alpha,\beta)}{ER_{H}(\alpha,\beta)+ER_{SM}(\alpha,\beta)+ER_{IM}(\alpha,\beta)}$$ and $$RREV_{IM}(\alpha,\beta)=\frac{ER_{IM}(\alpha,\beta)}{ER_{H}(\alpha,\beta)+ER_{SM}(\alpha,\beta)+ER_{IM}(\alpha,\beta)}.$$ 

We will give a lower bound of $ER_{IM}(\alpha,\beta)$ that can be proved to be strictly greater than $ER_{SM}(\alpha,\beta)$ to show that $RREV_{IM}(\alpha,\beta)>RREV_{SM}(\alpha,\beta)$ always holds.

We first represent $ER_{SM}(\alpha,\beta)$ as $h_1(\alpha,\beta)\cdot \pi_{(0,0)}$. From the stationary distribution of $Markov(\alpha,\beta)$, we obtain the following equations:
\begin{eqnarray*}
\left\{ \begin{array}{ll}
\pi_{(1,0)} = \alpha \pi_{(0,0)}; \\
\pi_{(0,1)}=\beta \pi_{(0,0)};\\
\pi_{(1,0')}=\beta \pi_{(1,0)}+\alpha\pi_{(0,1)};\\
\pi_{(0',0)}=(1-\alpha-\beta)\pi_{(1,0)};\\
\pi_{(0,0')}=(1-\alpha-\beta)\pi_{(0,1)};\\
\pi_{(i+1,0)} = \frac{\alpha}{(1-\alpha)} \pi_{(i,0)}, ~\forall i\geq 1;
\end{array}
\right.
\Longrightarrow~~
\left\{ \begin{array}{l}
\pi_{(1,0)} = \alpha \pi_{(0,0)}; \\
\pi_{(0,1)}=\beta \pi_{(0,0)};\\
\pi_{(1,0')}=2\alpha\beta\pi_{(0,0)};\\
\pi_{(0',0)}=\alpha(1-\alpha-\beta)\pi_{(0,0)};\\
\pi_{(0,0')}=\beta(1-\alpha-\beta)\pi_{(0,0)};\\
\pi_{(2,0)} = \frac{\alpha^2}{1-\alpha}\pi_{(0,0)};\\
\sum_{i=3}^{\infty} \pi_{(i,0)} = \frac{\alpha^3}{(1-\alpha)(1-2\alpha)} \cdot \pi_{(0,0)}.
\end{array}
\right.
\end{eqnarray*}

Then we have 
\begin{eqnarray*}
	&&ER_{SM}(\alpha,\beta) = \sum_{s\in\bbS}\sum_{\tilde{s}\in\bbS} Pr[s,\tilde{s}]\cdot r_{SM}[s,\tilde{s}]\cdot\pi_s\\
	&=& (\pi_{(0',0)}+\pi_{(1,0')})\cdot 1 \cdot \frac{1+3\alpha-\beta}{2} + \pi_{(0,0')}\cdot 1\cdot \alpha + \pi_{(2,0)}\cdot (1-\alpha)\cdot 2 + \sum_{i=3}^{\infty}\pi_{(i,0)}\cdot(1-\alpha)\cdot 1 \\
	&=&
	h_1(\alpha,\beta)\cdot \pi_{(0,0)},
\end{eqnarray*}
where $h_1(\alpha,\beta)\coloneqq \alpha(1-\alpha+\beta)\cdot \frac{1+3\alpha-\beta}{2}+\alpha^2(1-\alpha-\beta)+2\alpha^2+\frac{\alpha^3}{1-2\alpha}$.

We then give a lower bound of $ER_{IM}(\alpha,\beta)$, denoted by $ER_{IM}^*(\alpha,\beta)$, which could be written as $h_2(\alpha,\beta)\cdot \pi_{(0,0)}$.

From the stationary distribution of $Markov(\alpha,\beta)$, we obtain the following equations:
\begin{eqnarray*}
\left\{ \begin{array}{ll}
% \pi_{(1,0)} = \alpha \pi_{(0,0)}; & \\
% \pi_{(0,1)} = \beta \pi_{(0,0)}; &\\
\beta \pi_{(0,1)} = (1-\beta) (\pi_{(0',2)} + \pi_{(0,2)}) + \alpha \sum_{i \geq 1} \pi_{(i,2)}; &\\
\beta \sum_i \pi_{(i,j)} = (1-\beta) (\pi_{(0',j+1)} + \pi_{(0,j+1)}) + \alpha \sum_{i \geq 1} \pi_{(i,j+1)}, ~\forall j\geq 2.
\end{array}
\right.
\end{eqnarray*}

Note that $\alpha < \frac{1}{2} < 1-\beta$. We have

\begin{equation*}
    \begin{array}{rl}
    \beta \pi_{(0,1)} &= (1-\beta) (\pi_{(0',2)} + \pi_{(0,2)}) + \alpha \sum_{i \geq 1} \pi_{(i,2)} < (1-\beta) \sum_i \pi_{(i,2)} \\
    \beta \sum_i \pi_{(i,2)} &= (1-\beta) \pi_{(0',3)} + (1-\beta) \pi_{(0,3)} + \alpha \sum_{i \geq 1} \pi_{(i,3)} < (1-\beta) \sum_i \pi_{(i,3)}  \\
    \cdots&=\quad \cdots  \\
    \beta \sum_i \pi_{(i,j)} &= (1-\beta) \pi_{(0',j+1)} + (1-\beta) \pi_{(0,j+1)} + \alpha \sum_{i \geq 1} \pi_{(i,j+1)} < (1-\beta) \sum_i \pi_{(i,j+1)}\\
    \cdots&=\quad \cdots  \\
    \end{array} 
\end{equation*}

In other words,
\begin{equation}
    \begin{array}{l}
    \sum_i \pi_{(i,2)} > \frac{\beta}{1-\beta} \pi_{(0,1)} \\
    \sum_i \pi_{(i,3)} > \frac{\beta}{1-\beta} \sum_i \pi_{(i,2)} > (\frac{\beta}{1-\beta})^2  \pi_{(0,1)}   \\
    \cdots \\
    \sum_i \pi_{(i,j)} > (\frac{\beta}{1-\beta})^{j-1}  \pi_{(0,1)}   \\
    \cdots
    \end{array}
\end{equation}
Now
\begin{eqnarray*}
	&&ER_{IM}(\alpha,\beta) = \sum_{s\in\bbS}\sum_{\tilde{s}\in\bbS} Pr[s,\tilde{s}]\cdot r_{IM}[s,\tilde{s}]\cdot\pi_s\\
	&=& \pi_{(0',0)}\cdot \beta+\pi_{(1,0')}\cdot \frac{1-\alpha+3\beta}{2}+\pi_{(0,0')}\cdot \frac{1+3\beta}{2} + \left[(\pi_{(0',2)}+\pi_{(0,2)})\cdot (1-\beta)+\sum_{i\geq 1}\pi_{(i,2)}\cdot \alpha\right]\cdot 2\\
	&&  + \sum_{j\geq 3} \left[(\pi_{(0',j)}+\pi_{(0,j)})\cdot (1-\beta)+\sum_{i\geq 1}\pi_{(i,j)}\cdot \alpha\right]\\
    &=&\pi_{(0',0)}\cdot \beta+\pi_{(1,0')}\cdot \frac{1-\alpha+3\beta}{2}+\pi_{(0,0')}\cdot \frac{1+3\beta}{2} + \pi_{(0,1)}\cdot \beta\cdot 2\\
    &&  + \sum_{j\geq 3} \left[(\pi_{(0',j)}+\pi_{(0,j)})\cdot (1-\beta)+\sum_{i\geq 1}\pi_{(i,j)}\cdot \alpha\right]\\
    &>&\pi_{(0',0)}\cdot \beta+\pi_{(1,0')}\cdot \frac{1-\alpha+3\beta}{2}+\pi_{(0,0')}\cdot \frac{1+3\beta}{2} + \pi_{(0,1)}\cdot 2\beta+\pi_{(0,1)}\cdot \frac{\beta^2}{1-2\beta}.
\end{eqnarray*}

%So $ER_{IM}^*(\alpha,\beta) = h_2(\alpha,\beta)\pi_{(0,0)}$, where $$h_2(\alpha,\beta)=\frac{\alpha\beta(3+\beta-3\alpha)}{2}+\frac{\alpha\beta(1-\alpha+3\beta)}{2}+\frac{\beta(1-\alpha-\beta)(1+3\beta)}{2}+\frac{(2-3\beta)\beta^2}{1-2\beta}.$$

We define $ER_{IM}^*(\alpha,\beta)$ to be the last formula, \textit{i.e.}, 

\begin{equation*}
	ER_{IM}^*(\alpha,\beta) \coloneqq\pi_{(0',0)}\cdot \beta+\pi_{(1,0')}\cdot \frac{1-\alpha+3\beta}{2}+\pi_{(0,0')}\cdot \frac{1+3\beta}{2} + \pi_{(0,1)}\cdot 2\beta+\pi_{(0,1)}\cdot \frac{\beta^2}{1-2\beta}.
\end{equation*}
So $ER_{IM}^*(\alpha,\beta) = h_2(\alpha,\beta)\pi_{(0,0)}$, where $$h_2(\alpha,\beta)=\alpha\beta(1-\alpha-\beta)+2\alpha\beta\cdot\frac{1-\alpha+3\beta}{2}+\beta(1-\alpha-\beta)\cdot\frac{1+3\beta}{2}+2\beta^2+\frac{\beta^3}{1-2\beta}.$$

It is easy to verify that given $\alpha=\beta$, $h_2(\alpha,\beta)>h_1(\alpha,\beta)$ holds, so we have
\[ER_{IM}(\alpha,\beta)>ER_{IM}^*(\alpha,\beta)>ER_{SM}(\alpha,\beta),\] \textit{i.e.}, \[RREV_{IM}(\alpha,\beta)>RREV_{SM}(\alpha,\beta).\]

This completes the first part of our proof.

Now suppose that the stationary distribution of $Markov(\alpha,\beta)$ does not exist. Notice that our $Markov(\alpha,\beta)$ is irreducible, so every state in $Markov(\alpha,\beta)$ is not positive recurrent, \textit{i.e.}, the expected recurrence time of every state is infinite. If the first step of $Markov(\alpha,\beta)$ is transferring from $(0,0)$ to $(1,0)$, then in expectation, after finite steps the Markov Reward Process will return to $(0,0)$ because $\alpha<1/2<1-\alpha$. On the other hand, if the first step of $Markov(\alpha,\beta)$ is transferring from $(0,0)$ to $(0,1)$, then in expectation, the Markov Reward Process will have infinite steps in the states $\bigcup_{x}\bigcup_{y\geq 2}(x,y)$ because $(0,0)$ is not positive recurrent. During the infinite steps the honest pool and the selfish pool have no reward but there are about $\alpha*T$ rewards for the insightful pool, where $T$ is the steps. As $T\rightarrow \infty$, $RREV_{IM}(\alpha,\beta)\rightarrow 1$. This completes the second part of our proof.

\end{proof}

\section{The Mining Game and Equilibria}\label{sec: Mining Game and Equilibria}
In this section, we consider the scenario where all $n$ mining pools are strategic and study its Nash equilibrium. It is worth noting that insightful mining is a well-defined strategy and can be adopted directly. If there is no selfish pool in the system, insightful mining will look the same as the selfish mining. Then we consider the scenario where each pool can choose to follow the Bitcoin protocol truthfully or take the insightful mining strategy. We will formally define its strategy space in Section~\ref{sec: Strategy Space}, analyze the utility functions in Section~\ref{sec: Expected Reward Functions}, and characterize the Nash equilibrium in Section~\ref{sec:Characterization}.

\subsection{Strategy Space}\label{sec: Strategy Space}

There are $n$ mining pools and we denote it by $[n]\coloneqq\{1,\cdots,n\}$. The fraction of their hashing power is denoted by $\{m_1, \cdots, m_n\}$ and we have $\sum_{i=1}^n m_i = 1$. Each pool $i$ will infiltrate undercover miners into all other pools 
% (\textit{i.e.}, $\forall j \neq i$) 
to monitor their real-time state, namely, whether a certain pool is mining selfishly and if any, how many blocks are hidden. As a result, each pool $i$ could adopt the insightful mining strategy.

In the mining game, each pool has two strategies: \textit{refined honest mining} and \textit{insightful mining}, denoted by \textit{RHonest} and \textit{Insightful} respectively. The insightful mining strategy is exactly the same as we proposed before, while the refined honest mining is a slightly modified version of the standard mining strategy. Specifically, the refined honest mining requires the pool to mine after the longest public chain, and to publish its newly-generated block immediately. If someone hides the block, each pool could detect it through the spy therein. Then when facing two branches of the same length, the pool who adopts the refined honest mining strategy shall clearly follow the honest branch, instead of uniformly choosing one of them. 

%Each pool will plant an undercover miner into all other pools to detect if they are hiding blocks.  

%\textit{Case 1: a mining pool that adopts the honest mining strategy generates the first block.} This pool will broadcast it immediately and all other pools will accept this newly generated block, reach consensus and compete for the next block.

%\textit{Case 2: a mining pool that adopts the insightful mining strategy generates the first block.}

\subsection{Expected Reward Functions}\label{sec: Expected Reward Functions}
In this section, we give the formula of the expected reward function $ER_i(x_1,\cdots,x_n)$ of each pool $i$ under the pure strategy profile $(x_1,\cdots,x_n)\in\{RHonest, Insightful\}^n$.

\begin{proposition}\label{proposition: ER}
For an $n$-player mining game $(m_1,\cdots,m_n)$, let $(x_1,\cdots,x_n)$ be a (pure) strategy profile. Let $c$ be a value depending on $(m_1,\cdots,m_n)$ and $(x_1,\cdots,x_n)$.\footnote{We note that $c$ will not affect the  calculation of a pool's relative revenue in the subsequent section.} Let $Q\subseteq [n]$ be the set of pools that adopt $Insightful$ strategy. Then we have
\begin{eqnarray}
ER_i(x_1,\cdots,x_n)=\left\{ \begin{array}{ll}
c\cdot(f(m_i)+m_i\cdot\sum_{j\in Q}2m_j(1-m_j)), & i\in Q;\\
c\cdot(m_i+m_i\cdot\sum_{j\in Q}2m_j(1-m_j)), & i\not\in Q,\\
\end{array}
\right.
\end{eqnarray}
where $f(y)\coloneqq y^2\cdot (2-3y)/(1-2y)$.
\end{proposition}
\begin{proof}
Suppose that $Q=\{q_1,\cdots,q_k\}$ is the set of insightful pools and their hashing power is $\{m_{q_1},\cdots,m_{q_k}\}$. Other pools (\textit{i.e.}, $[n]\setminus Q$) adopt the refined honest mining strategy. Figure.~\ref{fig:game} illustrates the state transitions in the game. The top state $0$ is the initial state, meaning that there is a general consensus among all pools about the longest chain. For others, a state is composed of two parts, a symbol $y\in\mathbb{N^+}\cup \{0'\}$ and a pool label $(q_k)$, meaning that the insightful pool $q_k$ hides $y$ blocks while others reveal all their blocks. In particular, $0'(q_k)$ represents that there are two branches of the same length 1, in which one is held by $q_k$ and the other is held by a pool other than $q_k$.

We first analyze the expected reward of an insightful pool. 
For ease of notation, we focus on the pool $q_k \in Q$ (the right part of Fig.~\ref{fig:game}). The cases for the other insightful pools are
analogous and can be handled via an appropriate relabeling.
With probability $m_{q_k}$, the insightful pool $q_k$ finds a block and keeps it secretly. Then all other pools $[n]\setminus\{q_k\}$ will detect this behavior and mine on the other side. Denote $\pi_s$ as the stationary probability of the state $s$. We obtain the following equations
\begin{eqnarray*}
\left\{ \begin{array}{l}
m_{q_k} \pi_0 = \pi_{0'(q_k)} + (1-m_{q_k}) \pi_{2(q_k)};  \\
(1-m_{q_k}) \pi_{1(q_k)} = \pi_{0'(q_k)};   \\
m_{q_k} \pi_{i(q_k)} = (1-m_{q_k}) \pi_{i+1(q_k)},~\forall i \geq 1;
\end{array}
\right.
\Longrightarrow~
\left\{ \begin{array}{l}
\pi_{0'(q_k)} = m_{q_k} (1-m_{q_k}) \pi_0;   \\
\pi_{i(q_k)} = m_{q_k} \left(\frac{m_{q_k}}{1-m_{q_k}}\right)^{i-1} \pi_0,~\forall i \geq 1.
\end{array}
\right.
\end{eqnarray*}
Further, we calculate $q_k$'s expected revenue from this branch by
\begin{equation*}
    \pi_{0'(q_k)} \cdot m_{q_k} \cdot 2 + \pi_{2(q_k)} \cdot (1-m_{q_k}) \cdot 2 + \sum_{i=3}^\infty \pi_{i(q_k)} \cdot (1-m_{q_k}) \cdot 1 = \pi_0 \cdot \left[2m_{q_k}^2(1-m_{q_k}) + \frac{m_{q_k}^2(2-3m_{q_k})}{1-2m_{q_k}}\right].
\end{equation*}
In addition, the expected revenue of others as a whole from this branch is 
\begin{equation*}
    \pi_{0'(q_k)} \cdot (1-m_{q_k}) \cdot 2 = \pi_0 \cdot 2m_{q_k}(1-m_{q_k})^2,
\end{equation*}
since each pool $i \neq q_k$ could obtain $\pi_0 \cdot 2m_{q_k}(1-m_{q_k})^2 \cdot \frac{m_i}{1-m_{q_k}} = \pi_0 \cdot 2m_i m_{q_k} (1-m_{q_k})$ in expectation. Similarly, $q_k$ can also get revenue from other insightful $q_j$ ($q_j \in Q$ and $q_j \neq q_k$) to obtain $\pi_0 \cdot 2m_{q_k} m_{q_j} (1-m_{q_j})$. In total, an insightful pool $q_k$'s expected revenue is $$\pi_0 \cdot \left[\frac{m_{q_k}^2(2-3m_{q_k})}{1-2m_{q_k}} + 2m_{q_k} \sum_{q_j\in Q} m_{q_j}(1-m_{q_j})\right].$$
%\textbf{Case 1: $i\not\in S$.} First, the mining pool $i$ could generate the first block with probability $m_i$. Then it will broadcast this block immediately and get the revenue of one block. So the expected reward is $m_i$.

Second, we analyze the expected revenue of a pool $i \not\in Q$ that adopts the $RHonest$ strategy, which contains two parts. On the one hand, the pool $i$ may win some reward when the insightful pools mine the first block. Denote the pool $i$'s hashing power as $m_i$. As mentioned above, its expected revenue from this case is $\pi_0 \cdot [2m_i \sum_{q_j\in Q} m_{q_j}(1-m_{q_j})]$. On the other hand, with probability $m_i$, it mines a block, reveals it immediately, and gets the block's reward. Thus, in expectation its revenue 
from this case is $\pi_0 \cdot m_i$. In total, the pool $i$'s expected revenue is \[\pi_0 \cdot \left[m_i + 2m_i \sum_{j\in Q} m_j(1-m_j)\right].\]

So $c$ appeared in Proposition~\ref{proposition: ER} is defined as $\pi_0$, which depends on $m_1,\cdots,m_n$ and $(x_1,\cdots,x_n)$.

This finishes the proof.
\end{proof}

\begin{comment}
$q_1(0)$, $q_1(0')$, $q_1(1)$.

$s_1(0)$, $s_1(0')$, $s_1(1)$.

$s_k(0)$, $s_k(0')$, $s_k(1)$.

$0(s_1),0'(s_1)$, $1(s_1)$.

$0(s_k),0'(s_k)$, $1(s_k)$.

$\pi_{0(s_k)},\pi_{0'(s_k)}$, $\pi_{1(s_k)}$.
\end{comment}

\begin{figure}
\begin{center}
    \includegraphics[width=0.6\textwidth]{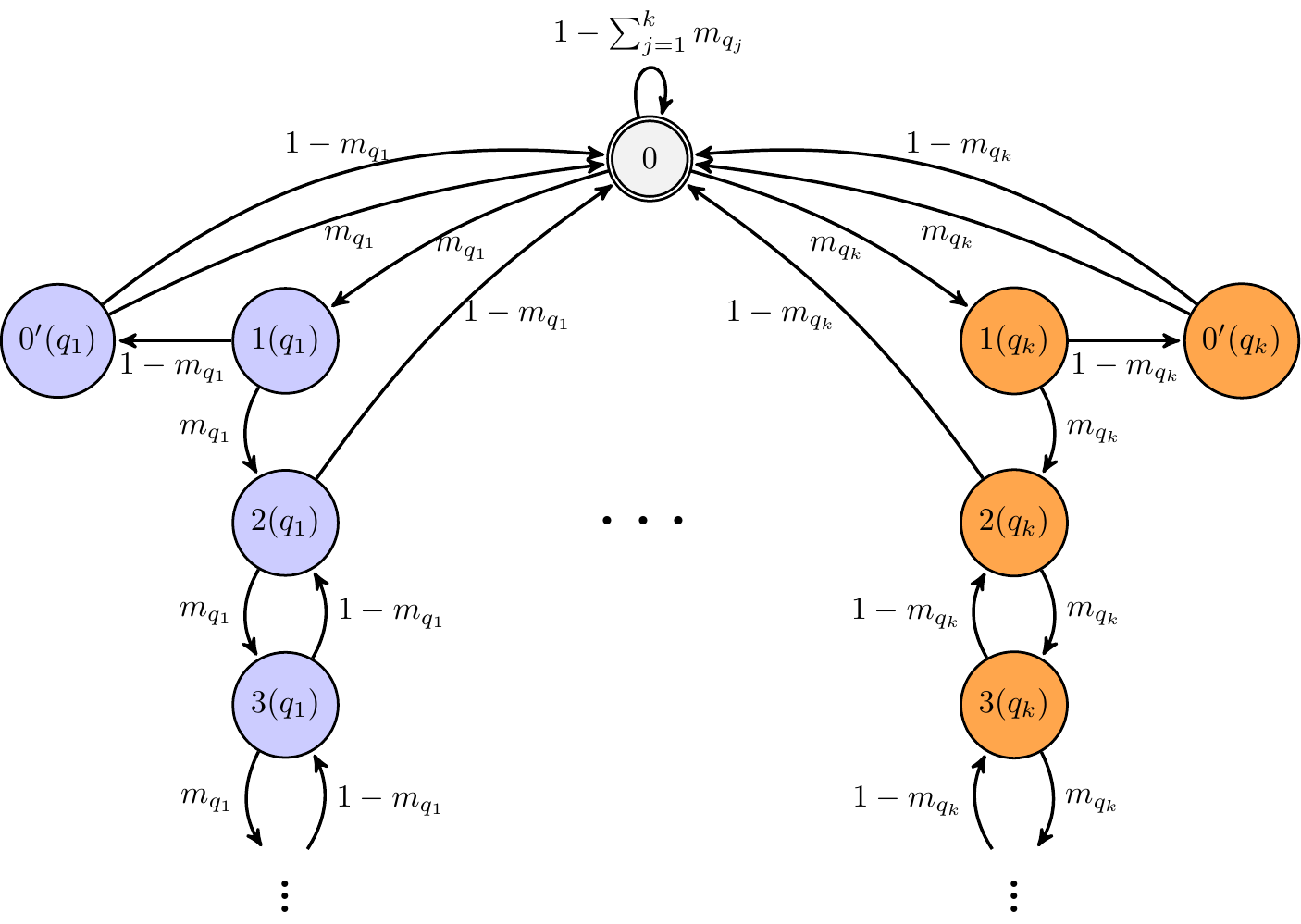}
    \caption{The state transitions in the mining game. Circles represent states, edges represent transitions, and the expression beside an edge represents the corresponding transition probability.}
    \label{fig:game}
    \end{center}
\end{figure}

\subsection{Equilibria Characterization}\label{sec:Characterization}

We denote the strategy profile other than a mining pool $i$ by $x_{-i}$. 

We prove the following characterization theorem of pure Nash equilibria:

\begin{theorem}\label{thm:equilibria}
For an $n$-player mining game $(m_1,\cdots,m_n)$ with $m_1\geq \cdots\geq m_n$, there are three types of pure Nash equilibrium $(x_1,\cdots,x_n)$, where
\begin{enumerate}
    \item $(x_1=\cdots x_n=RHonest)$ is a Nash equilibrium if and only if $m_1\leq 1/3$;
    \item $(x_1=Insightful,x_2=\cdots x_n=RHonest)$ is a Nash equilibrium if and only if $m_1\geq 1/3$ and $m_2\leq g(m_1)$;
    \item $(x_1=x_2=Insightful,x_3=\cdots x_n=RHonest)$ is a Nash equilibrium if and only if $m_1\geq 1/3$ and $m_2\geq g(m_1)$,
\end{enumerate}
where $g(y)\coloneqq \frac{-y^3+2y^2+y-1}{2y^2+4y-3}$.
\end{theorem}

Theorem~\ref{thm:equilibria} has the following three corollaries.

\begin{corollary}\label{coro: pure}
Every $n$-player mining game $(m_1,\cdots,m_n)$ has a pure Nash equilibrium.
\end{corollary}

\begin{corollary}\label{coro: 1/3 threshold}
For an $n$-player mining game $(m_1,\cdots,m_n)$, $(RHonest,\cdots,RHonest)$ is a Nash equilibrium if $m_1\leq 1/3$.
\end{corollary}

\begin{corollary}\label{coro: at most two}
For every $n$-player mining game $(m_1,\cdots,m_n)$, there is an equilibrium with at most two insightful pools.
\end{corollary}

\begin{proof}[Proof of Theorem~\ref{thm:equilibria} (1)]

Let $x=(H,\cdots,H)$ be the strategy profile such that all mining pools adopt $RHonest$ strategy. 

If all mining pools follow the strategy profile $x$, for a mining pool $i$, we have $ER_i(x)=m_i$ so that $RREV_i(x)=m_i$.

When the mining pool $i$ deviates to choose the $Insightful$ strategy, we have $ER_i(x_{-i},Insightful)=f(m_i)+2m_i^2(1-m_i)$ and $ER_j(x_{-i},Insightful)=m_j+2m_jm_i(1-m_i)$ for any other mining pool $j\neq i$. So for the relative revenue, we have
\[
RREV_i(x_{-i},Insightful)=\frac{f(m_i)+2m_i^2(1-m_i)}{1-m_i+f(m_i)+2m_i(1-m_i)}.
\]

Then $x$ is a Nash equilibrium if and only if every $i\in [n]$ satisfies $RREV_i(x_{-i},Insightful)\leq RREV_i(x)$. This, (given that $0< m_i\leq 1/2$), is equivalent to 

\begin{eqnarray*}
	&&\frac{f(m_i)+2m_i^2(1-m_i)}{1-m_i+f(m_i)+2m_i(1-m_i)}\leq m_i, ~~~\forall i\in[n]\\
	&\Longleftrightarrow& f(m_i)\leq m_i, ~~~\forall i\in[n]\\
	&\Longleftrightarrow& m_i\leq 1/3, ~~~\forall i\in[n]
\end{eqnarray*}

Recall that $m_1\geq \cdots\geq m_n$.
This finishes the proof.

\end{proof}

\begin{proof}[Proof of Theorem~\ref{thm:equilibria} (2)]

Let the strategy profile $x$ be $(Insightful,RHonest,\cdots,RHonest)$, i.e., the mining pool 1 adopts the $Insightful$ strategy and all other mining pools adopt the $RHonest$ strategy.

If all mining pools follow the strategy profile $x$, then (from the previous proof) for the relative revenue, we have
\[
RREV_{1}(x)=\frac{f(m_{1})+2m_{1}^2(1-m_{1})}{1-m_{1}+f(m_{1})+2m_{1}(1-m_{1})}.
\]

From Theorem~\ref{thm:equilibria} (1), we know that $RREV_{1}(x)\geq RREV_{1}(x_{-1},RHonest)$ if and only if $m_{1}\geq 1/3$. This means that ``the mining pool 1 has no incentive to deviate if and only if $m_1\geq 1/3$''.

Now it suffices for us to prove ``the mining pool $i\neq 1$ has no incentive to deviate if and only if $m_i\leq g(m_1)$''.

On the one hand, notice that when all mining pools follow the strategy profile $x$, we have
\[
RREV_i(x)=\frac{m_i+m_i\cdot 2m_1(1-m_1)}{1-m_1+f(m_1)+2m_1(1-m_1)}.
\]

On the other hand, when the mining pool $i$ deviates to choose the $Insightful$ strategy, we have \[ER_i(x_{-i},Insightful)=f(m_i)+2m_i(m_1(1-m_1)+m_i(1-m_i))\] and for any other mining pool $j\neq i$, we have \[ER_j(x_{-i},Insightful)=m_j+2m_j(m_1(1-m_1)+m_i(1-m_i)).\] 

For the relative revenue, we have
\[
RREV_i(x_{-i},Insightful)=\frac{f(m_i)+2m_i(m_1(1-m_1)+m_i(1-m_i))}{f(m_1)+f(m_i)+1-m_1-m_i+2(1-m_1-m_i)(m_1(1-m_1)+m_i(1-m_i))}.
\]

Then \[RREV_i(x_{-i},Insightful)\leq RREV_i(x)\Longleftrightarrow m_i\leq \frac{-m_1^3+2m_1^2+m_1-1}{2m_1^2+4m_1-3},\] i.e., $m_i\leq g(m_1)$.
\end{proof}

Before proving Theorem~\ref{thm:equilibria} (3), we first note that, in the sense of proof, the following Claim~\ref{claim} is a more natural statement. However, we will show that Claim~\ref{claim} and Theorem~\ref{thm:equilibria} (3) are essentially equivalent. Since the statement of (3) in Theorem~\ref{thm:equilibria} is more intuitive given the statements of (1) and (2), we adopt it in our theorem statement. 

\begin{claim}\label{claim}
For an $n$-player mining game $(m_1,\cdots,m_n)$ with $m_1\geq\cdots\geq m_n$, the strategy profile $(x_1=x_2=Insightful,x_3=\cdots x_n=RHonest)$ is a Nash equilibrium if and only if $m_1\geq g(m_2)$ and $m_2\geq g(m_1)$.
\end{claim}
\begin{proof}[Proof of Claim~\ref{claim}]
Let the strategy profile $x$ be $(Insightful,Insightful,RHonest,\cdots,RHonest)$, \textit{i.e.}, the first two mining pools adopt the $Insightful$ strategy and all other mining pools adopt the $RHonest$ strategy.

To prove the claim, we will show that ``$m_1\geq g(m_2)$ and $m_2\geq g(m_1)$'' is equivalent to ``every mining pool $i\in[n]$ has no incentive to change its strategy''. From the proof of Theorem~\ref{thm:equilibria} (2), we know that $m_1\geq g(m_2)$ is equivalent to $RREV_1(x)\geq RREV_1(x_{-1},RHonest)$ and $m_2\geq g(m_1)$ is equivalent to $RREV_2(x)\geq RREV_2(x_{-2},RHonest)$, so it remains for us to show ``$m_1\geq g(m_2)$ and $m_2\geq g(m_1)$'' is equivalent to ``for each mining pool $i\in[n]\setminus\{1,2\}$, $RREV_i(x)\geq RREV_i(x_{-i},Insightful)$''.

For a mining pool $i\in\{3,\cdots,n\}$, we have 
\[
RREV_i(x)=\frac{m_i+2m_i(m_1(1-m_1)+m_2(1-m_2))}{f(m_1)+f(m_2)+1-m_1-m_2+2(1-m_1-m_2)(m_1(1-m_1)+m_2(1-m_2))}
\]
and
\[
RREV_i(x_{-i},Insightful)=\frac{f(m_i)+2m_i(\sum_{j\in\{1,2,i\}}m_j(1-m_j))}{1+\sum_{j\in\{1,2,i\}}(f(m_j)-m_j)+2(1-\sum_{j\in\{1,2,i\}}m_j)(\sum_{j\in\{1,2,i\}}m_j(1-m_j))}.
\]

Then we know that $RREV_i(x)\leq RREV_i(x_{-i},Insightful)$ if and only if
\[
m_i\leq \frac{1 - m_1 - 2 m_1^2 + m_1^3 - m_2 + 4 m_1^2 m_2 - 2 m_1^3 m_2 - 2 m_2^2 + 
   4 m_1 m_2^2 + m_2^3 - 2 m_1 m_2^3}{3 - 4 m_1 - 2 m_1^2 - 4 m_2 + 
   4 m_1 m_2 + 4 m_1^2 m_2 - 2 m_2^2 + 4 m_1 m_2^2},
\]
where the right expression is defined as $h(m_1,m_2)$.

In addition, we have
\[
m_1+m_2+h(m_1,m_2)\geq 1
\]
for all $m_1\geq g(m_2)$, $m_2\geq g(m_1)$ and $0\leq m_2\leq m_1\geq 1/2$.
This means $h(m_1,m_2)\geq 1-m_1-m_2\geq m_i$. So we have $RREV_i(x)\leq RREV_i(x_{-i},Insightful)$, and thus $x$ is a Nash equilibrium.

\end{proof}

\begin{proof}[Proof of Theorem~\ref{thm:equilibria} (3)] 
Given the Claim~\ref{claim}, it suffices for us to show the equivalence between ``$m_1\geq g(m_2)$ and $m_2\geq g(m_1)$'' and ``$m_1\geq 1/3$ and $m_2\geq g(m_1)$'', given $0\leq m_2\leq m_1\leq 1/2$.

First notice that, given $m_1\geq m_2$, when $m_1,m_2\in[0,1/2]$ and $m_2\geq g(m_1)$, we have $m_1\geq g(m_1)$. This will derive $m_1\geq 1/3$.

So we only need to prove, given $m_1\geq m_2$ and $m_1,m_2\in[0,1/2]$, when $m_1\geq 1/3$ and $m_2\geq g(m_1)$, we have $m_1\geq g(m_2)$. We will use two simple cases to prove it.

\noindent\textbf{Case 1: $m_1\geq 0.35$.} Since $g(m_2)\leq 0.35$ for $0\leq m_2\leq 1/2$, so our desired result holds.

\noindent\textbf{Case 2: $m_1\leq 0.35$.} Then we have $m_1\in[1/3,0.35]$, so we have $m_2\geq g(m_1)\geq g(0.35)\geq 0.33$.

Notice that $g(m_2)$ is decreasing on $[0.33,1/2]$. So to show $m_1\geq g(m_2)$, it suffices for us to have $m_1\geq g(g(m_1))$, which is true when $m_1$ is on $[1/3,1/2]$.

This finishes the proof.
\end{proof}

% Now we give the proof of \Cref{claim} and finish the whole proof.

\section{Simulation}\label{sec:simulation}
\subsection{Insightful Mining vs. Selfish Mining}
\label{sec: Insightful Mining vs. Selfish Mining}

%In this subsection, we conduct the simulation to explore the long-term revenues of the selfish pool and the insightful pool. In our simulation, the selfish pool adopts the selfish mining strategy, while the insightful pool follow our insightful mining strategy. Both pools have the same mining power. We make all agents in the system to random walk based on their strategy. The simulation ends after 2,000,000,000 steps. Then we calculate the revenue over that period. 

The simulations in this subsection evaluate the effectiveness of the insightful mining strategy. Three agents are considered, \textit{i.e.}, the honest pool, the selfish pool, and the insightful pool. Their interactions are simulated as a discrete-time random walk process. In each step, one of the pools generates a block with the probability proportional to its hashing power, and others respond according to their strategies. The simulation ends after 2,000,000,000 steps. Then we calculate their relative revenue during the process, which is defined as the number of blocks generated by a pool over the number of blocks on the main chain. 

\begin{figure}
    \centering
    \includegraphics[width=\textwidth]{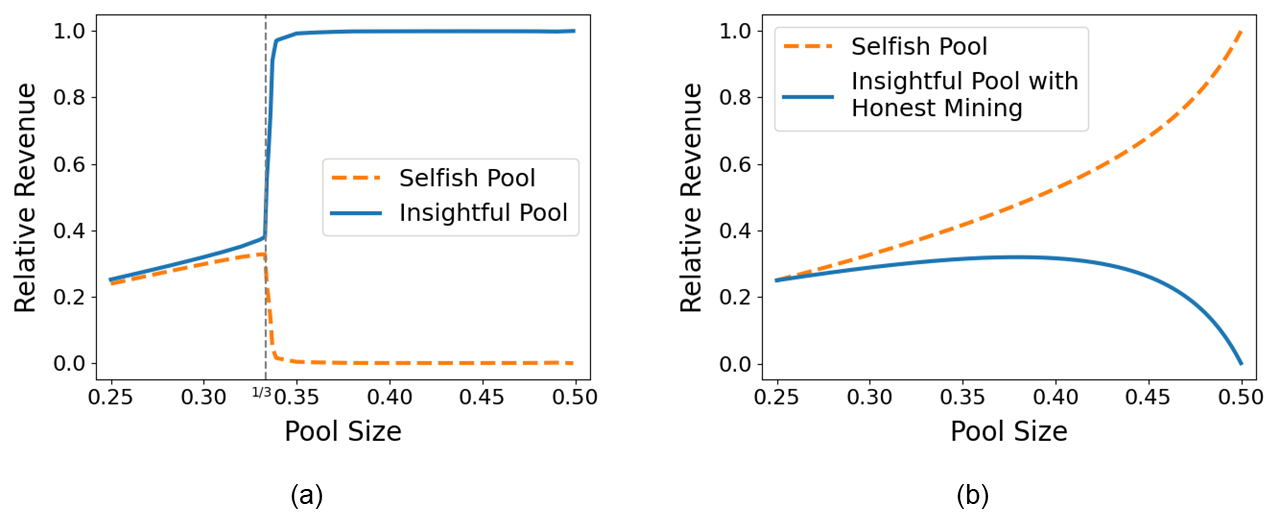}
    % \captionsetup{skip=7.5pt}
    \caption{Relative revenue of the selfish pool and the insightful pool with the same mining power. The selfish pool adopts the selfish mining strategy. (a) The insightful pool adopts the insightful mining strategy. (b) The insightful pool mines honestly.}
    \label{Fig:exp1} 
\end{figure}

Recall that $\alpha$ and $\beta$ are the fractions of hashing power of the selfish pool and the insightful pool, respectively. First, we focus on the scenario where the insightful pool and the selfish pool have the same hashing power, \textit{i.e.}, $\alpha=\beta$. Fig.~\ref{Fig:exp1}(a) visualizes the relative revenue of the insightful pool and the selfish pool when their hashing power belongs to $(0.25,0.5)$. As can be seen, the insightful pool can always gain more revenue than the selfish pool. It is exactly consistent with our theoretical result of Theorem~\ref{thm:domination}.
Surprisingly, if their hashing power is larger than $1/3$ (\textit{i.e.}, $\alpha=\beta>1/3$), the insightful pool can gain most of the revenue. For clear comparison, we also show their relative revenue under the circumstance that the insightful pool mines honestly in Fig.~\ref{Fig:exp1}(b). As mentioned in the Introduction, the insightful pool therefore will suffer heavy losses, which grow rapidly with the pool size increasing. Comparing Fig.~\ref{Fig:exp1}(a) and \ref{Fig:exp1}(b) shows that the insightful mining strategy dramatically help the pool turn things around when facing selfish mining. 

\begin{figure}
    \centering
    \includegraphics[width=0.5\textwidth]{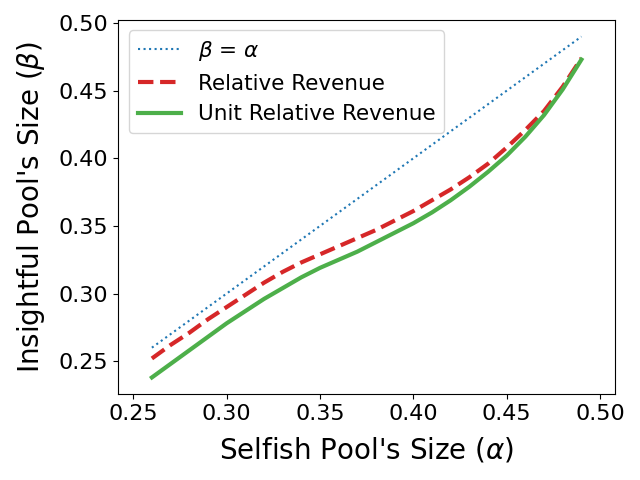}
    \caption{Threshold of the insightful pool's size, above which it could obtain more relative revenue or unit relative revenue than the selfish pool.} 
    \label{Fig:threshold}
\end{figure}

Then we explore the scenario where $\alpha>\beta$, to consider whether less hashing power can also enable the insightful pool to earn more. Here, two definitions of the `more revenue' are studied. One is the aforementioned relative revenue, which corresponds to the dashed line in Fig.~\ref{Fig:threshold}. It demonstrates the threshold above which $RREV_{IM}(\alpha,\beta)>RREV_{SM}(\alpha,\beta)$. The other is the unit relative revenue. The solid line in Fig.~\ref{Fig:threshold} represents the corresponding threshold, above which we have $\frac{RREV_{IM}(\alpha,\beta)}{\beta} > \frac{RREV_{SM}(\alpha,\beta)}{\alpha}$. This curve is below the latter. Both curves have similar trends and they are all below the line of $\beta=\alpha$. It provides compelling evidence that with insightful mining strategy, less computing power can also yield more revenue.

%Insightful Mining vs. Optimal Selfish Mining 

\subsection{Optimal Insightful Mining}
\label{sec: Optimal Insightful Mining}

\begin{figure}
    \centering
    \includegraphics[width=0.6\textwidth]{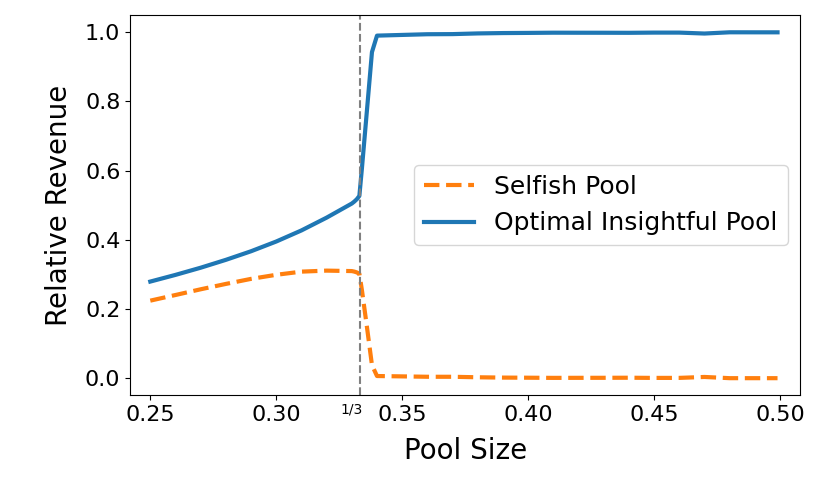}
    \caption{Relative revenue of the selfish pool and the insightful pool under the optimal insightful mining strategy. Both pools have the same mining power.}
    \label{fig:exp2} 
\end{figure}
In this subsection, we explore a more powerful insightful mining strategy by modeling the insightful pool's as a decision problem and solve it by the Markov Decision Process (MDP). First, we gives the description of the problem.

\noindent\textbf{Action Space.}
The insightful pool can take one of the following actions. \textit{(1) adopt to honest:} the insightful pool abandons its private fork to mine on the public chain. This action is always allowed. \textit{(2) adopt to selfish:} the insightful pool abandons its private fork to mine on the selfish branch. This action is feasible only when the selfish pool publishes its selfish branch triggered by the insightful pool publishing its private blocks. \textit{(3) override honest:} the insightful pool publishes just enough blocks on its private branch to overtake the public chain. This action is feasible only if the insightful branch is strictly longer than the public main chain. \textit{(4) override selfish:} the insightful pool publishes just enough blocks on its private branch to overtake the selfish branch. This action is feasible only if the insightful branch is strictly longer than the selfish branch. \textit{(5) wait:} the insightful pool continues to mine its private branch without publishing any block. This action is always feasible. \textit{(6) match:} the insightful pool publishes just enough blocks to equal the length of the longest public chain as the block on the longest public chain is being published, causing a fork.

\noindent\textbf{State Space.} The state of system is defined by a tuple of the form $(l_h, l_s, l_i, fork)$. The first three elements $l_h, l_s, l_i$ denote the length of of the honest branch, selfish branch, and insightful branch respectively. In particular, $l_s = -1$ indicates that the honest pool and the selfish pool mine on the same branch. The last element \textit{fork} contains three possible values, including \textit{relevant}, \textit{irrelevant}, and \textit{active}. It describes the event that just happened in the system: \textit{fork = relevant} when an honest pool just mined a block; \textit{fork = active} when the insightful pool just matched the honest branch; \textit{fork = irrelevant} otherwise.

% We should note that this representation is not perfect: for example, it cannot distinguish the state in which the honest miners mine in a chain mined by the agent with the state that the honest miners mine in a chain mined by the honest miners. But it is good enough.

\noindent\textbf{Transaction and Reward.} The transition matrix and reward matrix have been summarized in Table~\ref{tab2} in the~Appendix~\ref{sec:MDP Table}.

Similar to previous work~\cite{sapirshtein2016optimal}, we truncate this MDP into finite state by setting the maximal length of a branch as 50. 
%In this way, the optimal strategy solved by the simulation implies a lower bound of relative revenue of the truly optimal insightful mining. When improving the truncated threshold, the relative revenue of the insightful pool would be closer to be optimal.
As the objective function of this MDP is not linear, we apply the techniques presented in \cite{sapirshtein2016optimal} to solve this Average Reward Ratio Markov Decision Process (ARR-MDP). 

%As shown in the figure, the agent has significant superiority to selfish mining when their mining power are equal and greater than 0.25, and almost dominant when greater than 1/3.
Fig.~\ref{fig:exp2} shows the relative revenue of two pools under the optimal insightful strategy. It closely resembles Fig.~\ref{Fig:exp1}(a). Comparing two figures, we can find that when the mining power is less than 1/3, the optimal insightful mining strategy helps the insightful pool to gain more revenue than before. The detailed policy tells us that once the insightful branch is overridden during the mining competition, sticking to it for a while is a better choice than giving up immediately.

\subsection{Equilibria Visualization}
\label{sec: Equilibria Visualization}

\begin{figure}
    \begin{minipage}[t]{\linewidth}
        \centering
      \includegraphics[width=0.95\textwidth]{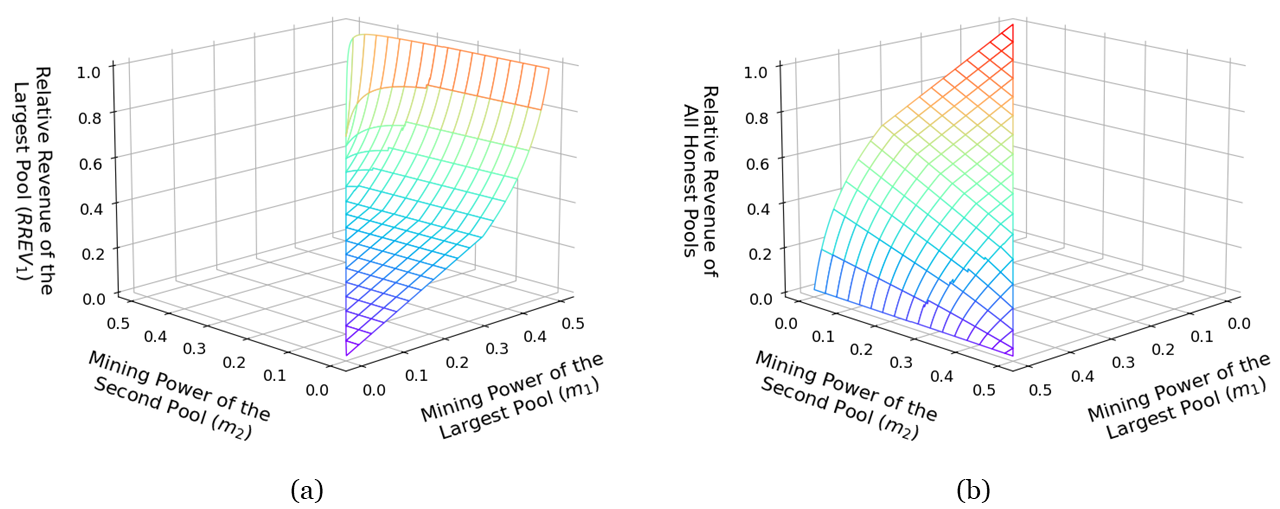}
    \end{minipage}
    \begin{minipage}[t]{\linewidth}
        \centering
     \includegraphics[width=0.95\textwidth]{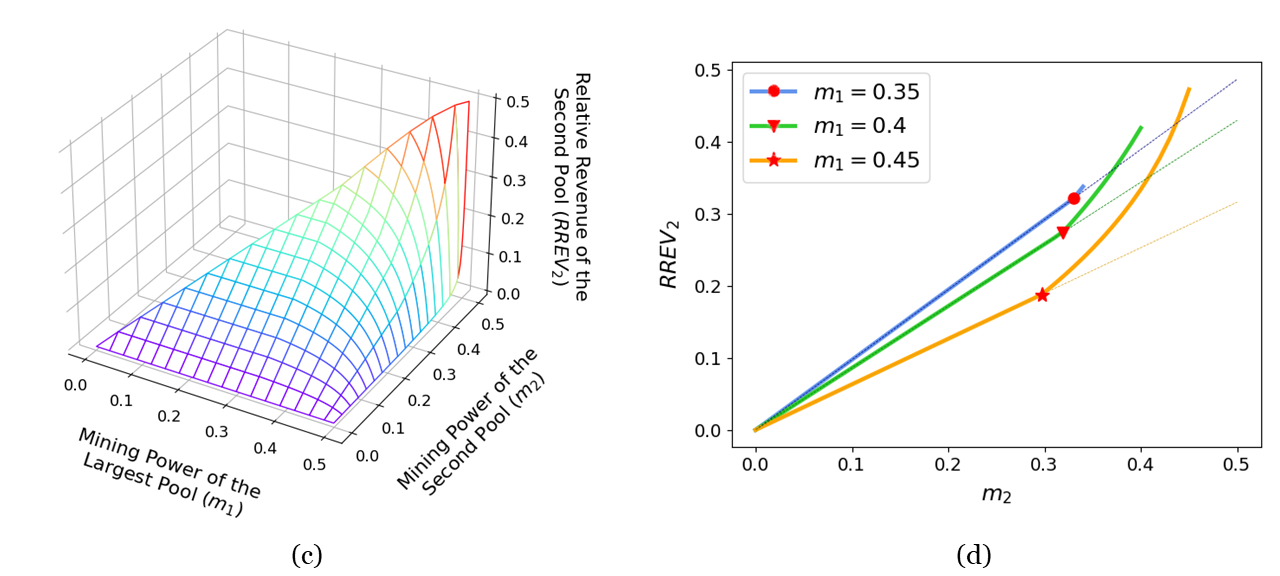}
    \end{minipage}
    % \captionsetup{skip=7.5pt}
    \caption{The relative revenue of pools under different distribution of the mining power.}
    \label{Fig:equilibrium} 
\end{figure}

In this subsection, we visualize the equilibrium results of Theorem \ref{thm:equilibria}. For the $n$ pools with ordered mining power $m_1 \geq \cdots \geq m_n$, the first player's equilibrium strategy only depends on its own mining power (\textit{i.e.}, $m_1$). For the players $i\in\{3,\cdots,n\}$, in any case, they will always be honest in equilibrium (referred to as \textit{honest pools} in the following). The strategy of the second pool, however, is jointly decided by both $m_1$ and $m_2$. For clear understanding, we visualize their equilibrium strategy and utility in Fig.~\ref{Fig:equilibrium}. 

For Fig.~\ref{Fig:equilibrium}(a)$\sim$(c), the X-axis is the mining power of the largest pool (\textit{i.e.}, $m_1$), which ranges in $(0,0.5)$. The Y-axis is the mining power of the second (largest) pool (\textit{i.e.}, $m_2$). Based on the assumption of $m_2 \leq m_1$, the possible values of $m_2$ are between 0 and $m_1$. The Z-axis of these three figures correspond to the relative revenue of the largest pool (\textit{i.e.}, $RREV_1$), the other pools, and the second pool (\textit{i.e.}, $RREV_2$) respectively. As can be seen, $RREV_1$ ($RREV_2$) increases with the increasing $m_2$, while the relative revenue of all honest pools decreases dramatically. An interesting observation is that, all these three figures surfaces have the same structure: when $m_1 \leq 1/3$, they are actually three planes. This is because in the circumstances, all pools will mine honestly in equilibrium, thus each play's revenue is proportional to its mining power. 

When $m_1 > 1/3$, as mentioned above, the second pool will take different strategies, depending on its $m_2$. For clear understanding, we further illustrate its relative revenue in the form of a 2D diagram in Fig.~\ref{Fig:equilibrium}(d). There are three lines, corresponding to three values of the largest pool's mining power (\textit{i.e.}, $m_1$=0.35, 0.4, and 0.45 respectively). For each line, the X-axis $m_2$ varies between 0 and $m_1$. As we can see, $RREV_2$ grows linearly first, then increases abruptly. The red mark on the line demonstrates the inflection point, after which the the second pool chooses to adopt the insightful mining strategy. By comparison, it can be found that as $m_1$ increases, the second pool will take insightful mining at a lower power level.

\section{Discussion}
\label{sec:discussion}
In blockchain, the action of planting a spy in the pool has been deeply discussed in the context of Block Withholding Attack~\cite{rosenfeld2011analysis,eyal2015miner}. In such an attack, the attacker infiltrates miners into opponent pools to reduce their revenue. The undercover miner sends only partial solutions (\textit{i.e.}, proofs of work) to the pool manager to share rewards. If it luckily finds a full solution which means a valid block, the undercover miner will discard the full proof of work directly, causing a loss to the victim pool. Our work explores, for the first time, the idea of spy in the selfish mining attack. It will shed new light upon the researchers in the field.% and further studies regarding the role of spy would be worthwhile.

Infiltrating spies dramatically expands the action spaces that a pool can take to counteract the selfish mining attack. Besides the insightful mining, other
% coping
strategies are worth exploring. Here, we roughly describe a potential idea. Recalling that %our insightful mining strategy only utilizes the number of hidden blocks, but 
the spy actually can extract the hash value of the latest hidden block from the new task issued by the pool manager. With this information, other pools can mine directly behind the latest block, although its full contents are not yet known.\footnote{Such an idea was also discussed in~\cite{sompolinsky2018bitcoin}. In that context, the strategic miner mines on a newly generated block directly even before it is validated. To avoid potential conflict, the miner can choose to embed no transaction in the block being mined, and just try to win the potential block rewards.}
By this strategy, all pools could follow the longest chain, which makes the selfish mining ineffective. In other words, keeping the block secretly for the selfish pool is equivalent to revealing it honestly, which extremely benefits the blockchain system. Nevertheless, such a strategy might not be the best choice for the strategic mining pools. Further research should be undertaken to investigate the optimal mining strategy. It is also worthwhile to extend the action of planting spies to other blockchain scenarios.

% Bibliography
\newpage
\bibliographystyle{ieeetr}

% Appendix
\newpage
\appendix
\section{Detailed Analysis of State Transitions}\label{sec:appendixA}
This section explains Table~\ref{tab:transition} in detail.

\begin{enumerate}
    \renewcommand{\labelenumi}{(\theenumi)}
    \item $Pr[(0,0),(0,0)] = 1 - \alpha - \beta$, $r[(0,0),(0,0)] = (1,0,0)$  \\
    This transition occurs when the honest pool finds the first block in this round of competition. The selfish pool and insightful pool will accept this block and mine after it. Thus, the honest pool obtains the revenue of one block.
    
    \item $Pr[(0,0),(1,0)] = \alpha$, $r[(0,0),(1,0)] = (0,0,0)$    \\
    This transition occurs when the selfish pool finds the first block and hides it immediately. In the state of $(1,0)$, the selfish pool faces the risk of being overtaken by others. The revenue of this block is uncertain, and will be determined in the later transition.
    
    %\item $Pr[(1,0),(0,0)] = 1-\alpha$, \\
    %$r[(1,0),(0,0)] = (\frac{(1-\alpha-\beta)(3-3\alpha+\beta)}{2},\frac{(1+3\alpha-\beta)(1-\alpha)}{2},\frac{\beta(3-3\alpha+\beta)}{2})$    \\
    %This transition contains two steps actually. Recall that the selfish pool hides a block and others have no block in the state of $(1,0)$. For the state transition from $(1,0)$ to $(0,0)$, there are two cases to discuss. 
    \item $Pr[(1,0),(0',0)] = 1-\alpha-\beta$, $r[(1,0),(0',0)] = (0,0,0)$    \\
    This transition occurs when the honest pool finds a block and publishes it directly. Then the selfish pool is forced to reveal its private block. It brings the system to state $(0',0)$, as Fig.~\ref{fig:omitted_states}(a) shows. In the state $(0',0)$, the honest branch and the selfish branch are of the same length (\textit{i.e.}, $l_h=l_s=1$). Faced with this situation, the honest pool uniformly follows one of them to mine the next block. The selfish pool persistently works on its branch, while the insightful pool goes ahead with the honest branch. Thus in state $(0',0)$, the total mining power on the honest branch and selfish branch is $\frac{1-\alpha-\beta}{2}+\beta$ and $\frac{1-\alpha-\beta}{2}+\alpha$ respectively. Both branches are possible to win. It will be determined by the next block. So the revenue of the block in this transition is uncertain and will be determined later. 
    
    \item  $Pr[(0',0),(0,0)] = 1$, $r[(0',0),(0,0)] = (\frac{3-3\alpha-\beta}{2},\frac{1+3\alpha-\beta}{2},\beta)$ \\
    This transition occurs when a block is generated, which breaks the tie in the state $(0',0)$ and also decides the winner. For this newly generated block, there are four possibilities. 
    \begin{enumerate}
        \item[(i)] With the probability $\frac{1-\alpha-\beta}{2}$, the honest pool finds a block on the honest branch. It invalidates the selfish branch and enables that the honest pool gets two blocks' revenue. 
        \item[(ii)] With the probability of $\beta$, the insightful pool finds a block on the honest branch. Then honest pool and the insightful pool obtain one block's revenue respectively. 
        \item[(iii)] With probability $\frac{1-\alpha-\beta}{2}$, the honest pool finds a block on the selfish branch. Then honest pool and the selfish pool obtains one block's revenue respectively. 
        \item[(iv)] With the probability of $\alpha$, the selfish pool finds a block on the its own branch and gets two blocks' revenue.
    \end{enumerate}
    To sum up, we could obtain that \\
    (a) the revenue of the honest pool is $\frac{1-\alpha-\beta}{2} \cdot 2 + \beta \cdot 1 + \frac{1-\alpha-\beta}{2} \cdot 1 = \frac{3-3\alpha-\beta}{2}$; \\
    (b) the revenue of the selfish pool is $\frac{1-\alpha-\beta}{2} \cdot 1 + \alpha \cdot 2 = \frac{1+3\alpha-\beta}{2}$; \\
    (c) the revenue of the insightful pool is $\beta \cdot 1 = \beta$.
    
    \item $Pr[(1,0),(1,0')] = \beta$, $r[(1,0),(1,0')] = (0,0,0)$ \\
    This transition occurs when the insightful pool finds a block and publishes it immediately. This event also promotes the selfish pool to reveal its hidden block. As a result, the system transfers to the state $(1,0')$, which we show in Fig.~\ref{fig:omitted_states}(b). Similar to what we discussed in the item (3), two branches in the state $(1,0')$ have the same length. The honest pool will uniformly choose one branch to work while the selfish pool and the insightful pool will mine after their own branches. Because which branch will win is still uncertain, we calculate the revenue of the block in this transition later.
    
    \item $Pr[(1,0'),(0,0)] = 1$, $r[(1,0'),(0,0)] = (1-\alpha-\beta,\frac{1+3\alpha-\beta}{2},\frac{1-\alpha+3\beta}{2})$    \\
    This transition occurs when a block is generated, which breaks the tie in the state $(1,0')$ and also decides the winner. Similar to the item (4), there are also four possibilities for the newly generated block.
    \begin{enumerate}
        \item[(i)] With probability $\frac{1-\alpha-\beta}{2}$, the honest pool finds a block on the selfish branch. Then honest pool and the selfish pool obtains one block's revenue respectively. 
        \item[(ii)] With the probability of $\alpha$, the selfish pool finds a block on the its own branch and gets two blocks' revenue.
        \item[(iii)] With the probability $\frac{1-\alpha-\beta}{2}$, the honest pool finds a block on the insightful branch. Then honest pool and the insightful pool obtain one block's revenue respectively.
        \item[(iv)] With the probability of $\beta$, the insightful pool finds a block on its insightful branch and thus gets two blocks' revenue.
    \end{enumerate}
    Then we could obtain that \\
    (a) the revenue of the honest pool is $\frac{1-\alpha-\beta}{2} \cdot 1 + \frac{1-\alpha-\beta}{2} \cdot 1 = 1-\alpha-\beta$; \\
    (b) the revenue of the selfish pool is $\frac{1-\alpha-\beta}{2} \cdot 1 + \alpha \cdot 2 = \frac{1+3\alpha-\beta}{2}$; \\
    (c) the revenue of the insightful pool is $\frac{1-\alpha-\beta}{2} \cdot 1 + \beta \cdot 2 = \frac{1-\alpha+3\beta}{2}$.

    %Fig.~\ref{fig:details}(a) shows the detailed state transitions from $(1,0)$ to $(0,0)$. Combining the above two cases, we could get $Pr[(1,0),(0,0)] = 1-\alpha$, $r[(1,0),(0,0)] = (\frac{(1-\alpha-\beta)(3-3\alpha+\beta)}{2},\frac{(1+3\alpha-\beta)(1-\alpha)}{2},\frac{\beta(3-3\alpha+\beta)}{2})$.
        
\begin{figure}[!t]
\includegraphics[width=\textwidth]{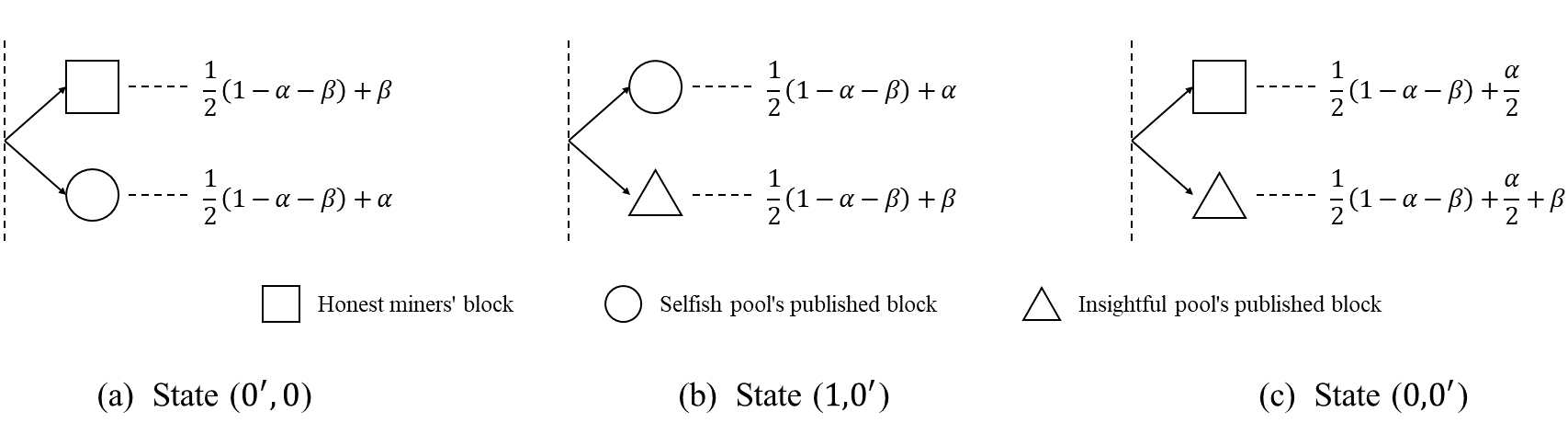}
\caption{Illustration of the states.} 
\label{fig:omitted_states}
\end{figure}

    \item $Pr[(x,0),(x+1,0)] = \alpha$, $r[(x,0),(x+1,0)] = (0,0,0)$, $\forall x \geq 1$    \\
    This transition occurs when the selfish pool finds a block and keeps it in secret. This action brings no revenue instantly. The reward of this block will be given to the selfish pool when it is released.
    
    \item $Pr[(2,0),(0,0)] = 1-\alpha$, $r[(2,0),(0,0)] = (0,2,0)$  \\
    This transition occurs when the insightful pool or the honest pool finds a block, which prompts the selfish pool to reveal its private branch of two blocks. It enables the selfish pool to obtain two blocks' revenue.
    
    \item $Pr[(x+1,0),(x,0)] = 1-\alpha$, $r[(x+1,0),(x,0)] = (0,1,0)$, $\forall x \geq 2$  \\
    This transition occurs when the insightful pool or the honest pool finds a block, and the selfish pool also reveals one private block at the same height. Because the selfish pool still has more blocks in secret, its selfish branch will become the main chain eventually. In this transition, the selfish pool obtain one block's revenue.
    
    \item $Pr[(0,0),(0,1)] = \beta$, $r[(0,0),(0,1)] = (0,0,0)$   \\
    This transition occurs when the insightful pool finds the first block and hides it immediately. In the state of $(0,1)$, the insightful pool faces the risk of being overtaken by others. The revenue of this block is uncertain, and will be determined in the later transition.
    
\begin{figure}[!t]
\includegraphics[width=\textwidth]{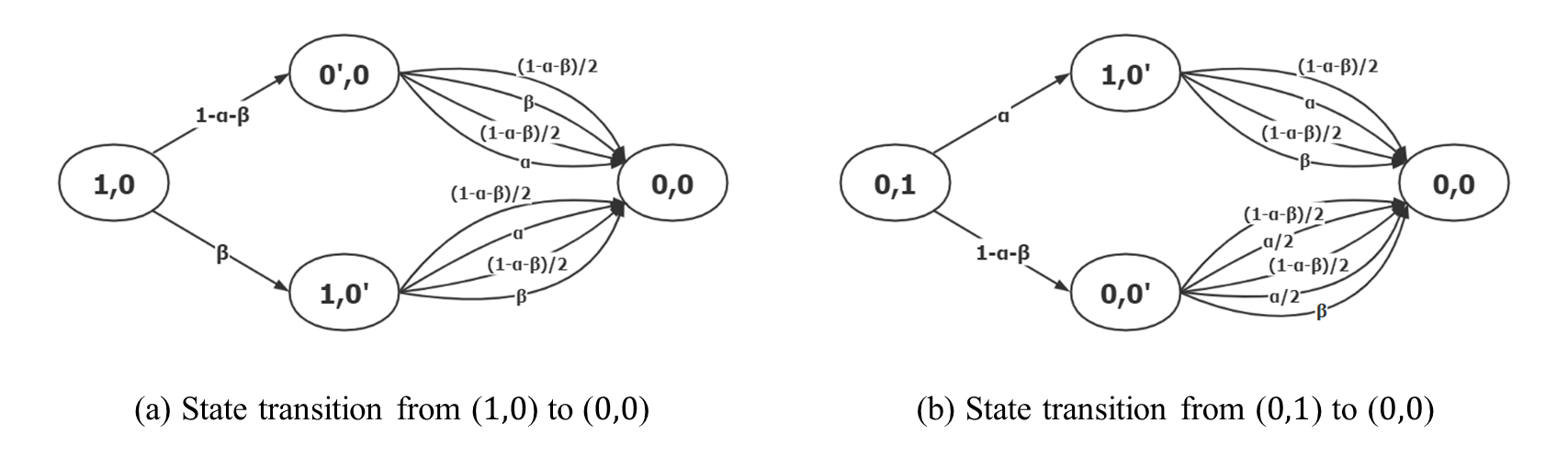}
\caption{Detailed state transitions.}
\label{fig:details}
\end{figure}

    %\item $Pr[(0,1),(0,0)] = 1-\beta$,  \\
    %$r[(0,1),(0,0)] = (\frac{(1-\alpha-\beta)(3-3\beta)}{2},\frac{\alpha(3+\alpha-3\beta)}{2},\frac{1+2\beta-\alpha^2-3\beta^2}{2})$   \\
    %This transition contains two steps actually. Recall that the insightful pool hides a block and others have no block in the state of $(0,1)$. For the state transition from $(0,1)$ to $(0,0)$, there are two cases to discuss.
    \item $Pr[(0,1),(1,0')] = \alpha$, $r[(0,1),(1,0')] = (0,0,0)$ \\
    This transition occurs when the selfish pool finds a block and also hides it. After observing this situation, the insightful pool reveals its hidden block and the selfish pool does the same subsequently. This brings the system to state $(1,0')$, which we show in Fig.~\ref{fig:omitted_states}(b). As discussed in the item (5), we will calculate the revenue of this block in the subsequent transition. 
    
    \item $Pr[(0,1),(0,0')] = 1-\alpha-\beta$, $r[(0,1),(0,0')] = (0,0,0)$    \\
    This transition occurs when the honest pool finds a block, promoting the insightful pool to reveal its hidden block. Then the system comes to a new state of $(0,0')$. As Fig.~\ref{fig:omitted_states}(c) shows, there are two branches, namely, the honest branch and the insightful branch. The selfish pool and the honest pool uniformly choose one of branches to follow, while the insightful pool sticks to its own branch. The revenue of the block will be determined later.
    
    \item $Pr[(0,0'),(0,0)]=1$, $r[(0,0'),(0,0)] = (\frac{3-2\alpha-3\beta}{2},\alpha,\frac{1+3\beta}{2})$  \\
    This transition occurs when a block is generated, which breaks the tie in the state $(0,0')$ and also decides the winner. There are five possibilities for the newly generated block.
    \begin{enumerate}
        \item[(i)] With the probability $\frac{1-\alpha-\beta}{2}$, the honest pool finds a block on the honest branch. It enables the honest pool to get two blocks' revenue. 
        \item[(ii)] With the probability of $\frac{\alpha}{2}$, the selfish pool finds a block on the honest branch. Then the honest pool and the selfish pool obtain one block's revenue respectively. 
        \item[(iii)] With the probability $\frac{1-\alpha-\beta}{2}$, the honest pool finds a block on the insightful branch. Then the insightful pool and the honest pool obtain one block's revenue respectively. 
        \item[(iv)] With the probability of $\frac{\alpha}{2}$, the selfish pool finds a block on the insightful branch. Then the insightful pool and the selfish pool obtain one block's revenue respectively. 
        \item[(v)] With the probability of $\beta$, the insightful pool finds a block on the its own branch and gets two blocks' revenue.
    \end{enumerate}
    Then we could obtain that   \\
    (a) the revenue of the honest pool is $\frac{1-\alpha-\beta}{2} \cdot 2 + \frac{\alpha}{2} \cdot 1 + \frac{1-\alpha-\beta}{2} \cdot 1 = \frac{3-2\alpha-3\beta}{2}$;    \\
    (b) the revenue of the selfish pool is $\frac{\alpha}{2} \cdot 1 + \frac{\alpha}{2} \cdot 1 = \alpha$; \\
    (c) the revenue of the insightful pool is $\frac{1-\alpha-\beta}{2} \cdot 1 + \frac{\alpha}{2} \cdot 1 + \beta \cdot 2 = \frac{1+3\beta}{2}$.
    % Fig.~\ref{fig:details}(b) shows the detailed state transitions from $(0,1)$ to $(0,0)$. By integrating these two cases, we could get $Pr[(0,1),(0,0)] = 1-\beta$, and $r[(0,1),(0,0)] = (\frac{(1-\alpha-\beta)(3-3\beta)}{2},\frac{\alpha(3+\alpha-3\beta)}{2},\frac{1+2\beta-\alpha^2-3\beta^2}{2})$.
    
    \item $Pr[(0,1),(0,2)] = \beta$, $r[(0,1),(0,2)]=(0,0,0)$   \\
    This transition occurs when the insightful pool finds a block and keeps it in secret. This action brings no revenue instantly. The reward of this block will be given to the insightful pool later.
    
    \item $Pr[(0,2),(0,0)] = 1-\beta$, $r[(0,2),(0,0)] = (0,0,2)$   \\
    This transition occurs when the selfish pool or the honest pool find a block. It prompts the insightful pool to reveal its private branch and therefore receive two blocks' revenue.
    
    \item $Pr[(x,y),(x,y+1)] = \beta$, $r[(x,y),(x,y+1)] = (0,0,0)$, $\forall x \in \{0'\}\bigcup \mathbb{N}$, $y \geq 2$    \\
    This transition occurs when the insightful pool finds a block and keeps it in secret. The reward of this block will be given to the insightful pool later.
    
    \item $Pr[(0,y),(0,y-1)] = 1-\alpha-\beta$, $r[(0,y),(0,y-1)] = (0,0,1)$, $\forall y \geq 3$    \\ 
    % 0,3->0,2
    This transition occurs when the honest pool finds a block. It reduces the insightful pool's leads by one. The insightful pool will not reveal its private block of the same height immediately, to avoid being detected by others. This insightful branch, however, will win as the main chain eventually. In order to facilitate the analysis, we decompose the whole chain rewards and calculate each block reward in the transition when the insightful pool's leads (\textit{i.e.}, $y$) decrease. Thus, in this transition, the insightful pool obtains one block's reward.
    
    \item $Pr[(x,y),(x+1,y-1)] = \alpha$, $r[(x,y),(x+1,y-1)] = (0,0,1)$, $\forall x \geq 0, y \geq 3$   \\
    % 0,3->1,2
    This transition occurs when the selfish pool finds a block. Then the insightful pool's leads decrease by one, but it will receive one block's revenue for this.
    
    \item $Pr[(1,y),(0',y)] = 1-\alpha-\beta$, $r[(1,y),(0',y)] = (0,0,0)$, $\forall y \geq 2$  \\
    % 1,2->0',2
    This transition occurs when the honest pool finds a block, forcing the selfish pool to reveal its only remaining block. This action has no effect on the insightful pool's leads which are no less than two, thus bringing no revenue.
    
    \item $Pr[(2,y),(0,y)] = 1-\alpha-\beta$, $r[(2,y),(0,y)] = (0,0,0)$, $\forall y \geq 2$  \\
    % 2,2->0,2
    This transition occurs when the honest pool finds a block. It forces the selfish pool to reveal its private branch of two blocks. However, it does not influence the insightful pool's leads. Thus, there is no revenue in this transition.
    
    \item $Pr[(x,y),(x-1,y)] = 1-\alpha-\beta$, $r[(x,y),(x-1,y)] = (0,0,0)$, $\forall x \geq 3, y \geq 2$    \\
    % 3,2->2,2
    This transition occurs when the honest pool finds a block, and the selfish pool reveals one private block at the same height. These is also no revenue.
    
    \item $Pr[(x,2),(0,0)] = \alpha$, $r[(x,2),(0,0)] = (0,0,2)$, $\forall x \geq 1$  \\
    % 1,2->0,0
    This transition occurs when the selfish pool finds a block. Then the insightful pool's leads reduce to one. It makes the insightful pool to reveal all its private blocks, which override other branches. The insightful pool obtains the last two blocks' revenue in this transition.
    
    \item $Pr[(0',2),(0,0)] = 1-\beta$, $r[(0',2),(0,0)] = (0,0,2)$     \\
    % Step1,2,3
    This transition occurs when the selfish pool or the honest pool finds a block. In the state of $(0',2)$, the selfish pool reveals all its blocks, and the honest pool has a public branch of the same length, and the insightful pool still has two blocks' leads. As the the selfish pool or the honest pool finds a new block, the tie is broken and the insightful pool's leads reduce to one. It prompts the insightful pool to reveal all its private blocks, which override other branches. The insightful pool obtains the last two blocks' revenue in this transition.

    \item $Pr[(0',y),(0,y-1)] = 1-\beta$, $r[(0',y),(0,y-1)] = (0,0,1)$, $\forall y \geq 3$ \\
    This transition occurs when the selfish pool or the honest pool finds a block. This newly generated block break the tie of them in the state of $(0',y)$, decreases the insightful pool' lead by one. The insightful pool receives one block's revenue for this.
    
\end{enumerate}

% \section{Missing Proof from \Cref{sec:Characterization}}\label{sec: proof of claim}

\section{Summary of MDP}\label{sec:MDP Table}
What we really care is the relative revenue of the insightful pool. So we use $(reward_h + reward_s, reward_i)$ to represent the expected reward. Define $l_s^*\coloneqq l_h$ if $l_s = -1$ otherwise $l_s^*\coloneqq l_s$.
\renewcommand{\thefootnote}{\fnsymbol{footnote}}

% \usepackage{multirow}
\begin{table}[ht]
\caption{A description of the transition and reward matrices in the decision problem.}\label{tab2}
\begin{tabular}{|c|l|l|l|l|}
\hline
\multicolumn{1}{|l|}{State $\times$ Action} & State  & Probability  & Reward  & Condition                              \\ \hline
\multirow{3}{*} {\begin{tabular}[c]{@{}c@{}}$(l_h,l_s,l_i,\cdot)$,\\$adopt$\end{tabular}}         & $(0,l_s^*-l_h,0,irrelevant)$ & 1 & $(l_h, 0)$  & $l_i < l_h$ \\ \cline{2-5} 
    & $(0,0,0,irrelevant)$            & 1                    & $(l_s, 0)$   & $l_i \ge l_h$, $l_s = l_i + 1 \ge 2$ \\ \cline{2-5} 
                                                                                                   & $(0,l_s-l_i,0,irrelevant)$    & 1                    & $(l_i, 0)$   & $l_i \ge l_h, l_s \ge l_i + 2$     \\ \hline
\multicolumn{1}{|c|}{\begin{tabular}[c]{@{}l@{}}$(l_h,l_s,l_i,\cdot)$,\\  $override_s$\end{tabular}} & $(0,0,l_i-l_s-1,irrelevant)$  & 1                    & $(0, l_s+1)$ & $l_i > l_s$                            \\ \hline
\multirow{8}{*}{\begin{tabular}[c]{@{}c@{}}$(l_h,l_s,l_i,   fork)$,\\  $wait$\end{tabular}}          & $(l_h+1,l_h+1,l_i, relevant)$ & $\alpha$             & (0,0)        & $l_s = -1$                             \\ \cline{2-5}                & $(l_h,l_s+1,l_i, relevant)$   & $\alpha$             & (0,0)        & $l_s \not= -1$                         \\ \cline{2-5} 
               & $(l_h,l_s,l_i+1, active)$     & $\beta$              & (0,0)        & fork = active                          \\ \cline{2-5} 
                                                                                                   & $(l_h,l_s,l_i+1, irrelevant)$ & $\beta$              & (0,0)        & fork $\not=$ active                    \\ \cline{2-5} 
                                                                                                   & $(l_h+1,l_h+1,l_i,relevant)$  & $1-\alpha-\beta$     & (0,0)        & $l_s^* \le l_h$                        \\ \cline{2-5} 
                                                                                                   & $(l_s,l_s,l_i,relevant)$      & $1-\alpha-\beta$     & (0,0)        & $l_s = l_h + 2$                        \\ \cline{2-5} 
                                                                                                   & $(l_h+1,-1,l_i,relevant)$     & $1-\alpha-\beta$     & (0,0)        & $l_s = l_h + 1$                        \\ \cline{2-5} 
                                                                                                   & $(l_h+1,l_s,l_i,relevant)$    & $1-\alpha-\beta$     & (0,0)        & $l_s > l_h + 2$                        \\ \hline
\multirow{9}{*}{\begin{tabular}[c]{@{}c@{}}$(l_h,l_s,l_i,   fork)$, \\ $match$\footnotemark[2]\end{tabular}}         & $(l_h+1,l_h+1,l_i, relevant)$ & $\alpha$             & (0,0)        & $l_s = -1$                             \\ \cline{2-5} 
                                                                                                   & $(l_h,l_s+1,l_i, relevant)$   & $\alpha$             & (0,0)        & $l_s \not= -1$                         \\ \cline{2-5} 
                                                                                                   & $(l_h,l_s,l_i+1, active)$     & $\beta$              & (0,0)        & fork = active                          \\ \cline{2-5} 
                                                                                                   & $(l_h,l_s,l_i+1, irrelevant)$ & $\beta$              & (0,0)        & fork $\not=$ active                    \\ \cline{2-5} 
                                                                                                   & $(1,1,l_i-l_h,relevant)$      & $(1-\alpha-\beta)/2$ & $(0,l_h)$    & $l_s^* \le l_h$                        \\ \cline{2-5} 
                                                                                                   & $(l_h+1,l_h+1,l_i,relevant)$  & $(1-\alpha-\beta)/2$ & (0,0)        & $l_s^* \le l_h$                        \\ \cline{2-5} 
                                                                                                   & $(l_s,l_s,l_i,relevant)$      & $1-\alpha-\beta$     & (0,0)        & $l_s = l_h + 2$                        \\ \cline{2-5} 
                                                                                                   & $(l_h+1,-1,l_i,relevant)$     & $1-\alpha-\beta$     & (0,0)        & $l_s = l_h + 1$                        \\ \cline{2-5} 
                                                                                                   & $(l_h+1,l_s,l_i,relevant)$    & $1-\alpha-\beta$     & (0,0)        & $l_s > l_h + 2$                        \\ \hline
\end{tabular}
\begin{tablenotes}
\footnotesize
\item
\footnotemark[2] The action \textit{match} is feasible only when $fork \not= irrelevant$ and $l_i \ge l_h$.
\end{tablenotes}
\end{table}

\end{document}